\documentclass[a4paper,fleqn]{cas-dc}

\usepackage[numbers]{natbib}
\usepackage{multirow}
\numberwithin{equation}{section}
\def\tsc#1{\csdef{#1}{\textsc{\lowercase{#1}}\xspace}}
\tsc{WGM}
\tsc{QE}
\usepackage{xcolor}

\newtheorem{theorem}{Theorem}

\newtheorem{proposition}{\textbf{Proposition}}

\newdefinition{remark}{Remark}
\newproof{proof}{Proof}
\newproof{pot}{Proof of Theorem \ref{thm}}
\usepackage{float}

\begin{document}
\let\WriteBookmarks\relax
\def\floatpagepagefraction{1}
\def\textpagefraction{.001}
\let\printorcid\relax 
\shorttitle{Boundary observer-based control of high-dimensional semilinear heat equations}

\shortauthors{Wang \& Fridman}  

\title [mode = title]{Constructive boundary observer-based control of high-dimensional semilinear heat equations}  

\tnotemark[1] 

\tnotetext[1]{This work was supported by Azrieli International Postdoctoral Fellowship, Israel Science Foundation (grant no. 673/19), the ISF-NSFC joint research program (grant no. 3054/23),  Chana and Heinrich Manderman Chair on System Control at Tel Aviv University, and the State Scholarship Fund of China Scholarship Council (grant no. 202306030133).}

%

\author{Pengfei Wang}



\ead{wangpengfei1156@hotmail.com}



\affiliation{organization={School of Electrical \& Computer Engineering, Tel-Aviv University},
            city={Tel Aviv},
            postcode={6997801}, 
            country={Israel}}

\author{Emilia Fridman}


\ead{emilia@tauex.tau.ac.il}







\begin{abstract}
This paper presents a constructive finite-dimensional output-feedback design for semilinear $M$-dimensional ($M\geq 2$) heat equations with boundary actuation and sensing. 
A key challenge in high dimensions is the slower growth rate of the Laplacian eigenvalues. 
The novel features of our modal-decomposition-based design, which allows to enlarge Lipschitz constants, include a larger class of shape functions that may be distributed over a part of the boundary only,  the corresponding lifting transformation and the full-order controller gain found from the design LMIs. 
We further analyze the robustness of the closed-loop system with respect to either multiplicative noise (vanishing at the origin) or additive noise (persistent).  
Effective LMI conditions are provided for specifying the minimal observer dimension and maximal Lipschitz constants that preserve the stability (mean-square exponential stability for multiplicative noise and noise-to-state stability for additive noise). Numerical examples for 2D and 3D cases demonstrate the efficacy and advantages of our method.
\end{abstract}




\begin{keywords}
High dimensional PDEs \sep stochastic heat equation \sep nonlinearity \sep boundary control \sep Lyapunov method.
\end{keywords}

\maketitle


\section{Introduction}

The S-procedure, introduced by Yakubovich \cite{yakubovich1971s}, along with Linear Matrix Inequalities (LMIs) inspired by his works (e.g., \cite{Yakubovich1962}), have become fundamental tools in robust control for both finite and infinite-dimensional systems. In this paper, we exploit these tools for finite-dimensional output-feedback control of semilinear heat equations, avoiding spillover effects induced by nonlinearities and output residue, and simultaneously showing robustness with respect to stochastic perturbations.

Finite-dimensional output-feedback controllers for PDEs are attractive in practical applications.
Such controllers were designed by the modal decomposition method and have been extensively studied since the 1980s \cite{balas1988finite,curtain1982finite,harkort2011finite}. The recent paper \cite{katz2020constructive} suggested a constructive LMI-based finite-dimensional observer-based control for 1D linear parabolic PDEs. The latter result was extended to semilinear parabolic PDEs \cite{katz2022globalscl}. 
In \cite{wang2023constructive}, we suggested a constructive finite-dimensional output-feedback control for 1D linear parabolic stochastic PDEs, extending results of \cite{katz2020constructive} to the stochastic case.
Motivated by \cite{wang2023constructive},  the almost sure exponential stability was studied in 
\cite{shang2025finite}. 
However, the above results are confined to 1D PDEs whereas the controller and observer gains have a reduced number of non-zero entries that correspond to the number of the unstable mode. The latter may be restrictive.

Control of high-dimensional PDEs became an active research area, with promising applications in engineering.
Observer-based boundary control for high-dimensional parabolic PDEs was explored in \cite{jadachowski2015backstepping,vazquez2016explicit} by backstepping method and in \cite{feng2022boundary,meng2022boundary} by modal decomposition approach. Note that the observers in \cite{feng2022boundary,jadachowski2015backstepping,meng2022boundary,vazquez2016explicit} were designed in the form of PDEs. 
The finite-dimensional boundary state-feedback stabilization of 2D linear heat equations was studied in \cite{vazquez2025backstepping} by combining the backstepping with Fourier analysis.

For high-dimensional linear parabolic PDEs with bounded control and observation operators, finite-dimensional output-feedback was suggested in \cite{djebour2024observer,sakawa1983feedback} without providing constructive bounds on the observer dimension. 
In \cite{wang2024auto}, we suggested the first constructive finite-dimensional output-feedback design for 2D linear heat equations where at least one of the observation or control operators is bounded. 
In \cite{wang2025sampled}, we explored the boundary state-feedback control for 2D semilinear stochastic PDEs by extending the lifting transformation in \cite{karafyllis2021lyapunov} to the 2D case. 
 The finite-dimensional control for 2D and 3D linear heat equation under boundary actuation and point measurements was studied in \cite{lhachemi2023boundary}. However, the LMI conditions provided for determining the observer dimension were conservative, and the numerical example failed to demonstrate their feasibility. 
 Furthermore, in \cite{lhachemi2023boundary}, a scaling technique was used to guarantee the stability feasibility for large enough observer dimension. This technique cannot guarantee the robustness with respect to stochastic perturbations (see Remark \ref{remark9}). Consequently, the development of a constructive finite-dimensional output-feedback control with an efficient design for $M$-dimensional ($M\ge 2$) PDEs and the extension to  the stochastic case remain an open challenge. 
 A slower growth rate of the eigenvalues of the high-dimensional Laplacian complicates the design and the feasibility analysis.

This paper for the first time presents a constructive and effective finite-dimensional output-feedback design for semilinear $M$-dimensional heat equations under boundary actuation and sensing. The nonlinearities satisfy the global Lipschitz condition. 
To address the challenge of slower growth rate of the Laplacian eigenvalues of high dimensions, we provide a rigorous estimate for the output residue and employing the S-procedure for spillover avoidance. The novel features of our modal-decomposition-based design include a larger class of shape functions that may be distributed over a part of the boundary only, the corresponding lifting transformation and the full-order controller gain found from the design LMIs. 
Our design essentially enlarges the admissible Lipschitz constants of the nonlinearities. 
We further analyze the robustness of the closed-loop system against both multiplicative noise (vanishing at the origin) and additive noise (persistent). 
By applying a direct Lyapunov method and using It\^{o}'s formulas, we establish mean-square exponential stability for multiplicative noise and the noise-to-state stability (NSS) for additive noise. The latter has not been studied for stochastic PDEs and cannot be studied under the input-to-state stability (ISS) umbrella since the stochastic integral against Brownian motion has infinite variation, whereas the integral of a legitimate input for ISS must have finite variation (see \cite{mateos2014p}).
Numerical examples for 2D and 3D cases demonstrate the efficacy of our method. The main contribution of this paper is listed as follows: 
\begin{itemize}
  \item For the first time, we study the finite-dimensional observer-based control of $M$-dimensional ($M\ge 2$) PDEs under boundary actuation and sensing, providing constructive LMI conditions for finding observer dimension. This is also the first paper on the control of stochastic PDEs (for $M\ge 1$) with NSS analysis.
  \item Differently from \cite{barbu2013boundary,lhachemi2023boundary,meng2022boundary,munteanu2017stabilisation,wang2025sampled}, where shape functions are restricted to eigenfunctions distributed over the entire boundary, we introduce a larger class of shape functions that may be distributed over a part of the boundary making design more flexible. Moreover, the latter leads to the corresponding lifting transformation that significantly improves the upper bounds on the Lipschitz constants and noise intensities (even in the state-feedback boundary control). 
    \item Compared to  \cite{katz2020constructive,katz2022globalscl,wang2023constructive} for 1D case where the controller gain was designed from the comparatively unstable modes,  this paper considers a full-order controller gain found from the design LMIs, essentially improving the upper bounds on the noise intensities and Lipschitz constants.
\end{itemize}

Preliminary results on the mean-square exponential stabilization
of 2D linear parabolic stochastic PDEs under the boundary actuation and sensing via the shape functions from \cite{meng2022boundary} were reported in CDC 2024 \cite{pengfei2024CDC}.

\textit{Notations:}
Let $(\Omega,\mathcal{F},\mathbb{P})$ be a complete probability space with a filtration $\{\mathcal{F}_{t}\}_{t\geq 0}$ of increasing sub $\sigma$-fields of $\mathcal{F}$ and let $\mathbb{E}\{\cdot\}$ be the expectation operator.
For bounded domain $\mathcal{O}\subset \mathbb{R}^{M}$, denote by $L^{2}(\mathcal{O})$ the space of all square integrable functions $f: \mathcal{O}\rightarrow \mathbb{R}$ with inner product $\langle f,g \rangle_{\mathcal{O}}=\int_{\mathcal{O}}f(x)g(x)\mathrm{d}x$ and induced norm $\|f\|_{L^{2}}^{2}=\langle f,f \rangle_{\mathcal{O}}$.  $H^{1}(\mathcal{O})$ is the Sobolev space of functions $f:\mathcal{O}\longrightarrow \mathbb{R}$ with a square integrable weak derivative.
Let $\frac{\partial}{\partial \nu}$ be the normal derivative. 
The Euclidean norm is denoted by  $|\cdot|$. For  $P\in \mathbb{R}^{n\times n}$, $P>0$ means that $P$ is symmetric and positive definite. The symmetric elements of a symmetric matrix will be denoted by $*$. For $0<P\in \mathbb{R}^{n\times n}$ and  $x\in \mathbb{R}^{n}$, we write $|x|^{2}_{P}=x^{\mathrm{T}}Px$. Denote $\mathbb{N}$ by the set of positive integers. 
For $A\in \mathbb{R}^{n\times n}$, the operator norm of $A$, induced by $|\cdot|$, is denoted by $\| A \|$. 
Let $I$ be the identity matrix with appropriate dimension and Id be the identity operator in $L^2(\mathcal{O})$. 

\section{Mean-square exponential stability for multiplicative noise}
\subsection{System under consideration}
Consider the stochastic semilinear heat equation:
\begin{equation}\label{eq14neumannrd}
\setlength{\arraycolsep}{1pt}		\begin{array}{ll}
		\mathrm{d} z(x,t)=[\Delta z(x,t)+qz(x,t)+f(x,z(x,t))]\mathrm{d}t \\
		~~~~~~~~~~~ +g(x,z(x,t))\mathrm{d}\mathcal{W}(t), ~~ t\geq 0, ~x\in \mathcal{O}, \\
\frac{\partial z(x,t)}{\partial \nu}=u(x,t), x\in\Gamma_{1},~ z(x,t)=0, ~x\in\Gamma_{2}, \\
		z(x,0)=z_0(x),
	\end{array}
\end{equation}
where $q\in\mathbb{R}$ is a reaction coefficient, $z_0\in L^2(\mathcal{O})$ is the initial value, and $\Delta z(x,t)=\sum_{j=1}^{M}\partial_{x_j}^2 z(x,t)$ denotes the Laplacian. 
Similarly to \cite{meurer2011flatness,meurer2009trajectory}, we consider the rectangular computational domain $\mathcal{O}=[0,a_1]\times \dots \times [0,a_M] \subset\mathbb{R}^M$ ($M\ge 2$) and consider one surface $\Gamma_1=\{(x_1,\dots,x_{M-1},0), 0\leq x_k\leq a_k, k=1,\dots,M-1\}$ for the boundary input and $\Gamma_2= \partial \mathcal{O} \backslash \Gamma_1$ for homogeneous boundary conditions. 
In \eqref{eq14neumannrd},  the noise term $g(x,z)\mathrm{d}\mathcal{W}(t)$ can be interpreted as the random perturbations of $f(x,z)\mathrm{d}t$. Here  $\mathcal{W}(t)$ is a standard one-dimensional Brownian motion defined on $(\Omega, \mathcal{F}, \mathbb{P})$. 
The nonlinear functions $f,g: \mathcal{O} \times \mathbb{R}\rightarrow \mathbb{R}$ satisfy
\begin{equation}\label{sigma0}
	\begin{array}{ll}
f(x,0)= 0, ~~g(x,0)= 0, \\
  |f(x,z_{1})- f(x,z_{2})| \leq \sigma_{f} |z_{1}-z_{2}|,\\
  |g(x,z_{1})- g(x,z_{2})| \leq \sigma_{g} |z_{1}-z_{2}|,~z_{1},z_{2}\in \mathbb{R},
\end{array}
\end{equation}
for all $x\in \mathcal{O}$, with some Lipschitz constants $\sigma_{f},\sigma_{g}>0$.
We consider the boundary measurement: 
\begin{equation}\label{measurement}
	\begin{array}{ll}
			 y(t)= \mathrm{col}\{ \langle c_m, z(\cdot, t) \rangle_{\Gamma_1} \}_{m=1}^{d}, \\
 c_m \in L^{2}(\Gamma_{1}), m=1,\dots,d,
	\end{array}
\end{equation}
where $\langle c_m, z(\cdot, t) \rangle_{\Gamma_1}=\int_{\Gamma_1} c_m(x) z(x, t) \mathrm{d} x$.
For simplicity, we consider deterministic measurement. As in \cite{wang2024SICON,wang2023constructive}, our results can be easily extended to a noisy measurement.

\begin{remark}
Using the transformation (52)-(54) from \cite[Appendix A]{meurer2009trajectory}, the following general parabolic system can be converted into the standard form \eqref{eq14neumannrd}: 
\begin{equation*}
	\begin{array}{ll}
	\mathrm{d} z(x,t)=[	\sum_{j=1}^{M}\alpha_j(x_j)\partial_{x_j}^2z(x,t)+\sum_{j=1}^{M}\beta_j(x_j)\partial_{x_j}z(x,t) \\
		~~~~~~~~~ +q(x)z(x,t)+f(x,z(x,t))]\mathrm{d}t\\
		~~~~~~~~~ +g(x,z(x,t))\mathrm{d}\mathcal{W}(t), ~~ t\geq 0, ~x\in \mathcal{O}, \\
\frac{\partial z(x,t)}{\partial \nu}=u(x,t), x\in\Gamma_{1}, ~   z(x,t)=0, ~x\in\Gamma_{2}.
	\end{array}
\end{equation*}
For this transformation to be valid, the parameters must satisfy the following regularity and boundedness conditions: $\alpha_j \in \mathcal{C}^2([0, a_j])$; $\beta_j \in \mathcal{C}^1([0, a_j])$; and $0 < \alpha_j^l \leq \alpha_j(x_j) \leq \alpha_j^u$, $0\leq \beta_j(x_j) \leq \beta_j^u$ for all $x_j \in [0, a_j]$; and $q^l \leq q(x) \leq q^u$ for all $x \in \bar{\mathcal{O}}$. Here, $\alpha_j^l$, $\alpha_j^u$, $\beta_j^u$, $q^l$, and $q^u$ are given constants.
\end{remark}

\begin{remark}
The choice of a rectangular domain with control applied on a single surface is primarily for computational simplicity (see, e.g., \cite{meurer2011flatness,meurer2009trajectory}). 
This framework can be extended to multiple control surfaces by introducing additional lifting transformations, as demonstrated for the 1D case with bilateral controls in \cite{katz2022network}. Theoretically, the results in this paper apply to any bounded, open, connected set $\mathcal{O} \subset \mathbb{R}^M$, as shown in related works \cite{barbu2013boundary,meng2022boundary,lhachemi2023boundary,pengfei2024CDC}.
However, for general domains, obtaining explicit eigenvalues and eigenfunctions for the operator \eqref{operatorA} below is challenging,  making it difficult to verify the resulting LMI conditions in real applications. 
\end{remark}
\begin{remark}
In this paper, we consider spatially uniform white noise to avoid technical difficulties. We propose a nonlinear noise perturbation function $g(x, z)$ to describe how the noise depend both on space and the state. We leave the more general case of spatially dependent white noise as a topic for future research. 
\end{remark}

Let
\begin{equation}\label{operatorA}
 \begin{array}{ll}
		\mathcal{A}\phi= -\Delta \phi, ~\mathcal{D}(\mathcal{A})= \{\phi| \phi\in H^{2}(\mathcal{O}) \cap H^{1}_{\partial}(\mathcal{O})\},\\
			H^{1}_{\partial}(\Omega)=\{\phi\in H^{1}(\mathcal{O})|\frac{\partial \phi(x)}{\partial \nu}=0 ~\mathrm{for}~ x\in \Gamma_1,\\
 ~~~~~~~~~~~~~~~~~~ \phi(x)=0 ~\mathrm{for}~ x\in \Gamma_2 \}.
	\end{array}
\end{equation}
Following \cite[Appendix II]{tucsnak2009observation},  we have the eigenvalues of $\mathcal{A}$:
\begin{equation}\label{eigenvalue1}
	\begin{array}{ll}
		\lambda_{\bf n}=\pi^2[ (\frac{{\bf n}_1}{a_1})^2 + \dots + (\frac{{\bf n}_{M-1}}{a_{M-1}})^2+ (\frac{{\bf n}_{M}-0.5}{a_{M}})^2  ] , \\
   {\bf n}=({\bf n}_1, \dots, {\bf n}_M ), ~~ {\bf n}_k\in \mathbb{N},
	\end{array}
\end{equation}
and the corresponding eigenfunctions: 
\begin{equation}\label{eigenfunction1}
	\begin{array}{ll}
		\phi_{\bf n}(x)= \sqrt{\frac{2}{a_M}}\cos\left (\frac{({\bf n}_M-\frac{1}{2})\pi x_M}{a_M} \right) \\
		 ~~~~~~~~~~~ \cdot \prod_{k=1}^{M-1}\left(\sqrt{\frac{2}{a_k}}\sin\left(\frac{{\bf n_k}\pi x_k}{a_k}\right)\right),\\
 x=(x_1,\dots,x_M)\in \mathcal{O}.   
	\end{array}
\end{equation}
We reorder the eigenvalues \eqref{eigenvalue1} to form a nondecreasing sequence $\{\lambda_{n}\}_{n=1}^{\infty}$ that satisfies
\begin{equation}\label{eigenvalue2}
	\lambda_{1} < \lambda_{2} \leq \dots \leq \lambda_{n} \leq \dots, ~~~
	\lim_{n\rightarrow\infty} \lambda_{n}=\infty,
\end{equation}
and denote the corresponding eigenfunctions \eqref{eigenfunction1} as $\{\phi_{n}\}_{n=1}^{\infty}$.
The eigenfunctions $\{\phi_{n}\}_{n=1}^{\infty}$ form a complete orthonormal system in $L^{2}(\mathcal{O})$.
For the eigenvalues \eqref{eigenvalue2}, it follows from \cite[Corollary 3.6.8]{tucsnak2009observation} that
\begin{equation}\label{george}
	\begin{array}{ll}
		\lambda_n=O(n^{\frac{2}{M}}), ~~n\to \infty.
	\end{array}
\end{equation}

Let $\delta>0$ be a desired decay rate. Since $\lim_{n\to \infty}\lambda_n=\infty$, there exists some $N_{0}\in \mathbb{N}$ such that
\begin{equation}\label{NNNed}
		-\lambda_{n}+q+\delta +\sigma_f  <0,~~n>N_{0},
\end{equation}
where $N_{0}$ denotes the number of unstable modes. This $N_0$ was chosen in \cite{wang2025sampled} for the 2D state-feedback control.
Let $N \in  \mathbb{N}$, $N \geq N_0$.
In this paper, the controller gain and observer dimension are determined by the first $N$ modes.

\begin{remark}
Differently from \cite{katz2020constructive,katz2022globalscl,wang2023constructive} for the 1D case where the controller gain was designed using the first $N_0$ modes,
 this paper designs the controller gain from the first $N$ (observer dimension) modes.
In the simulation from the conference version of this paper \cite{pengfei2024CDC}, we showed that the controller gain based on $N$ modes results in approximately 10 times larger noise intensity than the gain based on $N_0$ modes, as in \cite{katz2020constructive,katz2022globalscl,wang2023constructive}.
\end{remark}

Let $d$ be the maximum geometric multiplicity of $\{\lambda_{n}\}_{n=1}^{N_0}$. Clearly, $d\leq N_0$. 
We consider the boundary control input in the following form:  
\begin{equation}\label{controlshape}
	\begin{array}{ll}
u(x,t) = \sum_{i=1}^{d}b_i(x)u_i(t), ~x\in\Gamma_1,  
\end{array}
\end{equation}
where $b_i\in L^2(\tilde{\Gamma}_{i,1})$, $\tilde{\Gamma}_{i,1} \subset \Gamma_1$, $i=1,\dots,d$, and $b_i =0$ for $x \in \partial \tilde{\Gamma}_{i,1} $ and $ x \in \Gamma_1 \backslash \tilde{\Gamma}_{i,1}$.  
It means that each $b_i$ is supported over $\tilde{\Gamma}_{i,1} \subset \Gamma_1$ (see Fig. \ref{figplant} {\bf(b)} and {\bf(d)} for the 2D and 3D cases). 
Note that in \cite{meng2022boundary,meurer2011flatness,meurer2009trajectory,wang2025sampled,barbu2013boundary}, the shape functions are distributed over the entire $\Gamma_1$ (see Fig. \ref{figplant} {\bf(a)} and {\bf(c)}). This highlights the flexibility of our design.

  


\subsection{Lifting transformation and dynamic extension}\label{sec2.2}
In this section, we employ a lifting transformation to obtain an equivalent system with homogeneous boundary conditions. 
Given the shape functions $b_i$ ($i=1,\dots,d$),  
we let $ {\mathcal{O}}_i\subset  \mathcal{O}$ be such that $\tilde{\Gamma}_{i,1} \subset \partial  {\mathcal{O}}_i$. Let $\tilde{\Gamma}_{i,2}= \partial  {\mathcal{O}}_i \backslash \tilde{\Gamma}_{i,1} $. Let $\mu_i\neq \lambda_n$, $i=1,\dots,d$ be given.
We define a sequence of functions $\psi_{i}\in L^{2}(\mathcal{O}_i)$ satisfying the following elliptic equation (see \cite{feng2022boundary,meng2022boundary}):
\begin{equation}\label{eq16neumann22ard}
	\begin{split}
	&\Delta \psi_{i}(x)   = -\mu_{i} \psi_{i}(x), ~x\in  {\mathcal{O}}_i, \\
	&\frac{\partial\psi_{i}(x)}{\partial \nu}=  b_{i}(x), ~x\in \tilde{\Gamma}_{i,1},~~\psi_{i}(x)=0, ~x\in \tilde{\Gamma}_{i,2}. 
	\end{split}
\end{equation}
We extend the domain of $\psi_i$ to $\mathcal{O}$ by letting $\psi_i(x)=0$ for $x\in \bar{\mathcal{O}}\backslash  {\mathcal{O}}_i$. 

\begin{remark}
In \cite{barbu2013boundary,lhachemi2023boundary,meng2022boundary,munteanu2017stabilisation,wang2025sampled}, $\psi_i$ are constructed in the entire domain $\mathcal{O}$ (see Fig. \ref{figplant}{\bf(a)} for the 2D case and Fig. \ref{figplant}{\bf(c)} for the 3D case). In this paper, we suggest a larger class of $\psi_i$ that can be distributed over subdomains of $\mathcal{O}$ or on the entire $\mathcal{O}$ (see Fig. \ref{figplant}), leading to essentially larger Lipschitz constants (see Tables \ref{tab1} and \ref{table2} in Section \ref{example}).
\end{remark}

\begin{figure*}[hbt!]
\centerline{\includegraphics[width=12.5cm]{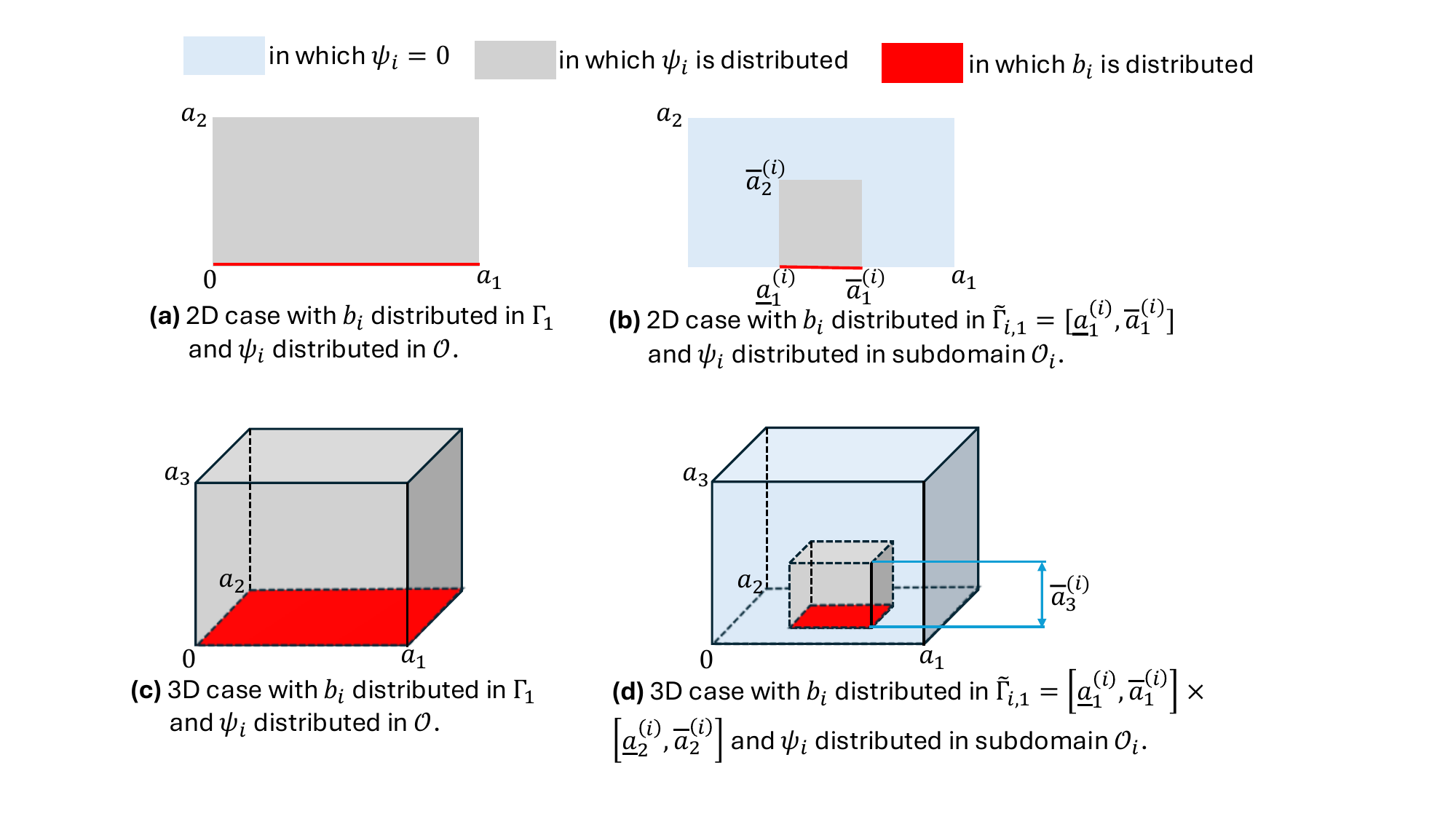}} 
\caption{ {\bf (a)} and {\bf (c)} show the distribution of $b_i$ over $\Gamma_1$ and $\psi_i$ over $\mathcal{O}$; {\bf (b)} and {\bf (d)} show the distribution of $b_i$ over $\tilde{\Gamma}_{i,1}\subseteq \Gamma_1$ and $\psi_i$ over $\mathcal{O}_i\subseteq \mathcal{O}$.}
\label{figplant}
\end{figure*}

Consider the following lifting transformation:  
\begin{equation}\label{eq17neumannrdaa}
	\begin{split}
&		w(\cdot,t)=z(\cdot,t)- \psi^{\mathrm{T}}(\cdot){\bf u}(t), ~\\
 &\psi=\mathrm{col}\{\psi_{i}\}_{i=1}^{d},~~{\bf u}(t)= \mathrm{col}\{u_{m}(t)\}_{m=1}^{d}.
	\end{split}
\end{equation}
Substituting \eqref{eq17neumannrdaa} into \eqref{eq14neumannrd} and using operator \eqref{operatorA}, we obtain  
\begin{equation}\label{dyeq18neumannrd}
\setlength{\arraycolsep}{0.01pt}		\begin{array}{ll}
		
		\mathrm{d}w(t)=[-\mathcal{A}w(t)+ qw(t)] \mathrm{d}t\\
	~~~~~~~~~ - \psi^{\mathrm{T}}(\cdot)[\mathrm{d} {\bf u}(t)- \Xi_{0}{\bf u}(t)\mathrm{d}t  ] \\
	~~~~~~~~~ + f (\cdot, w(t) + \psi^{\mathrm{T}}(\cdot)  {\bf u}(t))\mathrm{d}t \\
    ~~~~~~~ ~~ +g (\cdot, w(t) + \psi^{\mathrm{T}}(\cdot)  {\bf u}(t)) \mathrm{d}\mathcal{W}(t), \\
	\end{array}
\end{equation}
where $w(t)=w(\cdot,t)$, $\Xi_{0}=\mathrm{diag}\{-\mu_{1}+q,\dots,-\mu_{d}+q\}$.
We treat ${\bf u}(t)\in\mathbb{R}^{d}$ as an additional state variable subject to the dynamics:
\begin{equation}\label{eq19neumannrd}
	\mathrm{d} {\bf u}(t)=[\Xi_{0}{\bf u}(t)+{\bf v}(t)]\mathrm{d}t,~~~{\bf u}(0)=0,
\end{equation}
whereas ${\bf v}(t)\in\mathbb{R}^{d}$ is the new control input.
From \eqref{dyeq18neumannrd} and \eqref{eq19neumannrd}, we have the following equivalent system:
\begin{equation}\label{eq20neumannrd}
 \begin{array}{ll}
	\mathrm{d}w(t)=[-\mathcal{A}w(t)+ qw(t) - \psi^{\mathrm{T}}(\cdot){\bf v}(t)] \mathrm{d}t\\
	~~~~~~~~~ + f (\cdot, w(t) + \psi^{\mathrm{T}}(\cdot)  {\bf u}(t))\mathrm{d}t \\
    ~~~~~~~ ~~ +g (\cdot, w(t) + \psi^{\mathrm{T}}(\cdot)  {\bf u}(t)) \mathrm{d}\mathcal{W}(t),\\
    w(0)=z(\cdot,0). 
\end{array}
\end{equation}
 The boundary measurement becomes
\begin{equation}\label{measurementnewrd}
	\begin{array}{ll}
y(t)= \mathrm{col}\{\langle c_m, w(\cdot, t)  + \psi ^{\mathrm{T}}(\cdot){\bf u}(t)\rangle_{\Gamma_1}\}_{m=1}^{d}.
\end{array}
\end{equation}

\subsection{Finite-dimensional observer and controller}\label{sec2.3}
Denote by $\mathbb{P}_{N}$ the projector on the finite dimensional subspace $\mathcal{X}_{N}=\text{span}\{\phi_n\}_{n=1}^{N}$. We have $\mathcal{X}_{N}=\mathbb{P}_{N}L^2(\mathcal{O})$, and for any $h \in L^2(\mathcal{O})$, 
\begin{equation*}
	\begin{array}{ll}	\mathbb{P}_{N}h(\cdot)=\sum_{n=1}^{N} \langle h,\phi_{n} \rangle_{\mathcal{O}}\phi_n,\\
(\text{Id}-\mathbb{P}_{N})h=\sum_{n=N+1}^{\infty} \langle h,\phi_{n} \rangle_{\mathcal{O}}\phi_n.
	\end{array}
\end{equation*}

Let \(\mathcal{A}_N\) be the restriction of \(\mathcal{A}\) to the finite-dimensional subspace \(\mathcal{X}_N\), i.e., \(\mathcal{A}_N : \mathcal{X}_N \to \mathcal{X}_N\) is defined by \(\mathcal{A}_N h = \mathbb{P}_N(\mathcal{A}h)\) for \(h \in \mathcal{X}_N\). Similarly, let \(\mathcal{A}_{\text{tail}}\) be the restriction of \(\mathcal{A}\) to \(\mathcal{X}_{\text{tail}}\), defined by \(\mathcal{A}_{\text{tail}} h = (\text{Id} - \mathbb{P}_N)(\mathcal{A}h)\) for \(h \in \mathcal{X}_{\text{tail}} \cap \mathcal{D}(\mathcal{A})\). 
The operators leave their respective subspaces invariant: \(\mathcal{A}_N(\mathcal{X}_N) \subset \mathcal{X}_N\) and \(\mathcal{A}_{\text{tail}}(\mathcal{X}_{\text{tail}}) \subset \mathcal{X}_{\text{tail}}\). 
The spectra of these restricted operators are given by \(\sigma(\mathcal{A}_N) = \{\lambda_n\}_{n=1}^{N}\) and \(\sigma(\mathcal{A}_{\text{tail}}) = \{\lambda_n\}_{n=N+1}^{\infty}\) (see \cite{barbu2011internal,kato2013perturbation}). 
We shall write \eqref{eq20neumannrd} as
\begin{equation}\label{eq20neumannrdX0}
 \begin{array}{ll}
	\mathrm{d}X_{N}(\cdot,t)=[-\mathcal{A}_{N}X_{N}(\cdot,t)+ qX_{N}(\cdot,t) \\
	~~~~~ -\mathbb{P}_{N} \psi^{\mathrm{T}}(\cdot){\bf v}(t)] \mathrm{d}t\\
	~~~~~ + \mathbb{P}_{N}f (\cdot, w(t) + \psi^{\mathrm{T}}(\cdot)  {\bf u}(t))\mathrm{d}t \\
    ~~~~~ +\mathbb{P}_{N}g (\cdot, w(t) + \psi^{\mathrm{T}}(\cdot)  {\bf u}(t)) \mathrm{d}\mathcal{W}(t),\\
     X_{N}(\cdot,0)=\mathbb{P}_{N}w(0),
\end{array}
\end{equation}
and
\begin{equation}\label{eq20neumannrdtail}
 \begin{array}{ll}
	\mathrm{d}X_{\mathrm{tail}}(\cdot,t)=[-\mathcal{A}_{\mathrm{tail}}X_{\mathrm{tail}}(\cdot,t)+ qX_{\mathrm{tail}}(\cdot,t) \\
	~~~~~ -(\text{Id}-\mathbb{P}_{N}) \psi^{\mathrm{T}}(\cdot){\bf v}(t)] \mathrm{d}t\\
	~~~~~ + (\text{Id}-\mathbb{P}_{N})f (\cdot, w(t) + \psi^{\mathrm{T}}(\cdot)  {\bf u}(t))\mathrm{d}t \\
    ~~~~~ +(\text{Id}-\mathbb{P}_{N})g (\cdot, w(t) + \psi^{\mathrm{T}}(\cdot)  {\bf u}(t)) \mathrm{d}\mathcal{W}(t),\\
     X_{\mathrm{tail}}(\cdot,0)= (\text{Id}-\mathbb{P}_{N})w(0). 
\end{array}
\end{equation}
Recall that $X_N(\cdot,t)=\mathbb{P}_{N}w(\cdot,t)=\sum_{n=1}^{N} w_{n}(t)\phi_n$, where $w_{n}(t)=\langle w(\cdot,t),\phi_{n} \rangle_{\mathcal{O}}$.  We rewrite \eqref{eq20neumannrdX0} as
\begin{equation}\label{wnwnrd}
	\begin{array}{ll}
		\mathrm{d} w_{n}(t)= [(-\lambda_{n}+q)w_{n}(t) + f_n(t)  + {\bf b}_{n}{\bf v}(t)  ]\mathrm{d}t\\
		~~~~~~~~~~~ +g_n(t) \mathrm{d}\mathcal{W}(t),~  ~t\geq 0, \\
		w_{n}(0)=\langle w(\cdot,0), \phi_{n} \rangle_{\mathcal{O}}, ~~n=1,\dots,N,
	\end{array}
\end{equation}
where
\begin{equation*}
	\begin{array}{ll}
		{\bf b}_{n}= [\langle -\psi_{1}, \phi_{n}\rangle_{\mathcal{O}}, \cdots, \langle -\psi_{d}, \phi_{n}\rangle_{\mathcal{O}}   ], \\
f_n(t)=\langle f(\cdot,w(\cdot,t) + \psi^{\mathrm{T}}{\bf u}(t) ), \phi_n \rangle_{\mathcal{O}}, \\
g_n(t)=\langle g(\cdot,w(\cdot,t) + \psi^{\mathrm{T}}{\bf u}(t) ), \phi_n \rangle_{\mathcal{O}}.
	\end{array}
\end{equation*}
Following \cite{katz2020constructive},  we construct a $N$-dimensional observer of the form
\begin{equation}\label{observer22300rd}
  \begin{array}{ll}
 \hat{w}(x,t)=\sum_{n=1}^{N}\hat{w}_{n}(t)\phi_{n}(x), \\
 \mathrm{d} {\hat{w}}_{n}(t)=[(-\lambda_{n}+q)\hat{w}_{n}(t) + \hat{f}_n(t) + {\bf b}_{n}{\bf v}(t)]\mathrm{d}t \\
 ~~~~~~~~~~~ -l_{n} [\hat{y}(t) - y(t) ]\mathrm{d}t, ~~t>0, ~1\leq n\leq N, \\
  \hat{w}_{n}(0)=0,
  \end{array}
\end{equation}
where $y(t)$ is defined in \eqref{measurementnewrd}, $\{l_{n}\in \mathbb{R}^{1 \times d}\}_{n=1}^{N}$ are observer gains,
and  
\begin{equation*}
\begin{array}{ll}
 \hat{f}_n(t)=\langle f(\cdot,\hat{w}(\cdot,t) + \psi^{\mathrm{T}}{\bf u}(t) ), \phi_n \rangle_{\mathcal{O}},\\
 \hat{y}(t)= \mathrm{col}\{\langle c_m, \hat{w}(\cdot, t)  + \psi ^{\mathrm{T}}(\cdot){\bf u}(t)\rangle_{\Gamma_1}\}_{m=1}^{d}. 
\end{array}
\end{equation*}


Introduce the notations:  
\begin{equation}\label{notationsaa}
\setlength{\arraycolsep}{0.1pt}		\begin{array}{ll}
A_0=\mathrm{diag}\{-\lambda_{n}+q\}_{n=1}^{N_0}, ~ {\bf B}_0=[{\bf b}^{\mathrm{T}}_{1},\dots, {\bf b}^{\mathrm{T}}_{N_0}]^{\mathrm{T}},\\
	{\bf c}_n = [ {\bf c}_{1,n}, \dots, {\bf c}_{d,n}]^{\mathrm{T}},
 {\bf c}_{m,n}= \langle {c}_{m}, \phi_{n} \rangle_{\Gamma_1}, \\
 {\bf C}_{0}  = [{\bf c}_{1},\dots, {\bf c}_{N_0}].   
	\end{array}
\end{equation}
We rewrite $A_{0}, {\bf B}_{0}, {\bf C}_{0}$ as:  
\begin{equation}\label{A0A0A0dy}
\begin{array}{ll}
	A_{0}= \mathrm{diag}\{\tilde{A}_{1}, \dots, \tilde{A}_{p}\}, \\
	\tilde{A}_{j}=\mathrm{diag}\{-\lambda_{j}+q, \dots, -\lambda_{j}+q\}\in\mathbb{R}^{n_{j}\times n_{j}},\\
	\lambda_{k} \neq \lambda_{j} ~~\mathrm{iff}~~k\neq j, ~~ k,j=1,\dots,p, \\
{\bf B}_{0}=[B^{\mathrm{T}}_{1}, \dots, B^{\mathrm{T}}_{p}]^{\mathrm{T}},~~ B_{j}\in\mathbb{R}^{n_{j}\times d},\\
{\bf C}_{0}=[C_{1}, \dots, C_{p}],~~C_{j}\in\mathbb{R}^{d\times n_{j}},  
\end{array}	
\end{equation}
 where $n_{1},\dots, n_{p}$ are positive integers such that $n_{1}+\dots+n_{p}=N_{0}$. Clearly, $n_{j}\leq d$, $j=1,\dots,p$.
Assume  
\begin{equation}\label{eq20dy}
	\mathrm{rank}(C_{j})= n_{j},\mathrm{rank}(B_{j})= n_{j},~j=1,\dots,p.
\end{equation}
By Lemma 2.1 in \cite{wang2024auto},  under assumption \eqref{eq20dy},
the pair $(A_{0}, {\bf C}_{0})$ is observable and $(A_0,{\bf B}_0)$ is controllable. Hence, we can choose $l_{1},\dots, l_{N_0}$ such that
$L_0=[l_{1}^{\mathrm{T}},\dots, l_{N_0}^{\mathrm{T}}]^{\mathrm{T}}\in\mathbb{R}^{N_0\times d}$ satisfies the following Lyapunov inequality
\begin{equation}\label{observergainfulldy}
		P_{o}(A_0-L_0{\bf C}_0)+ (A_0-L_0{\bf C}_0)^{\mathrm{T}}P_{o}<-2\delta P_{o},
\end{equation}
where $0<P_{o}\in\mathbb{R}^{N_0 \times N_0}$.
Furthermore, as in \cite{katz2020constructive,sakawa1983feedback}, we choose $l_n=0_{1\times d}$ for $n = N_0 +1, \dots ,N$, to guarantee the LMI feasibility.

Let
$A=\mathrm{diag}\{-\lambda_{n}+q\}_{n=1}^{N}$, ${\bf B}=[{\bf b}^{\mathrm{T}}_{1},\dots, {\bf b}^{\mathrm{T}}_{N}]^{\mathrm{T}}$,  $\tilde{A}= \mathrm{diag}\{\Xi_{0}, A \}$, $\tilde{{\bf B}}=[I_{d}, {\bf B}^{\mathrm{T}} ]^{\mathrm{T}}$. Since $(A_0,{\bf B}_0)$ is controllable and $A_1$ is stable (see \eqref{NNNed}), 
the pair $(A, {\bf B})$ is stabilizable, which implies that $(\tilde{A},\tilde{{\bf B}})$ is stabilizable.
Let $K\in \mathbb{R}^{d\times (d+N)}$ be the controller gain (it will be found from the resulting LMIs).
We propose a controller of the form  
\begin{equation}\label{controllawrd}
	\begin{array}{ll}
		{\bf v}(t)=-K\hat{w}^{N}(t), \\
		\hat{w}^{N}=\mathrm{col}\{{\bf u},\hat{w}_{1},\dots,\hat{w}_{N}\}.
	\end{array}
\end{equation}

Well-posedness of \eqref{eq19neumannrd}, \eqref{eq20neumannrd}, \eqref{observer22300rd} subject to controller \eqref{controllawrd} can be obtained by arguments similar to \cite[Sec. 2.1.2]{wang2023constructive}: For any initial value $z_0\in  \mathcal{D}(\mathcal{A})$, system \eqref{eq19neumannrd}, \eqref{eq20neumannrd} with control input \eqref{observer22300rd},  \eqref{controllawrd} admits a unique strong solution
{$w \in L^{2}(\Omega; C([0,T]; L^{2}(\mathcal{O}) ) ) \cap L^2([0,T] \times \Omega;H^{1}(\mathcal{O}))$} for any {$T>0$}, such that $w(\cdot,t)\in \mathcal{D}(\mathcal{A})$, $t\geq 0$, almost surely.

\subsection{Mean-square $L^2$ exponential stability}\label{sec2.4}
Let $e_{n}(t)=w_{n}(t)-\hat{w}_{n}(t)$, $1\leq n \leq N $
be the estimation error. Denote by $e^{N}(t) = \mathrm{col}\{e_{n}(t)\}_{n=1}^{N}$.
The last term on the right-hand side of \eqref{observer22300rd} can be presented as
\begin{equation}\label{tail100rd}
	\begin{array}{ll}
		\hat{y}(t) - y(t) = -{\bf C}e^{N}(t) -\zeta(t) , ~{\bf C}  = [{\bf c}_{1},\dots, {\bf c}_{N}], \\
		\zeta(t)=\sum_{n=N+1}^{\infty}{\bf c}_{n}w_{n}(t),  
	\end{array}
\end{equation}
where ${\bf c}_n$ are defined in \eqref{notationsaa}. Note that the output residue $\zeta(t)$ appears in the finite-dimensional part of the closed-loop system (see \eqref{eqXX0000aa} and \eqref{eqXX0011aa} below). We use the Cauchy-Schwarz inequality for the estimate of $\zeta^2(t)$:
\begin{equation}\label{eq3333rd}
\begin{array}{ll}	
|\zeta(t)|^{2}
	\leq \kappa_{N}\sum_{n=N+1}^{\infty}\lambda_{n} w_{n}^{2}(t),
\end{array}
\end{equation}
where 
\begin{equation*}
	\begin{array}{ll}
		\kappa_N=\sum_{n=N+1}\frac{|{\bf c}_n|^2}{\lambda_n}  =\sum_{m=1}^{d}\sum_{n=N+1 }^{\infty} \frac{{\bf c}_{m,n}^{2}}{\lambda_{n}}. 
	\end{array}
\end{equation*}
We give the following estimate for $\kappa_N$, which is crucial for the feasibility analysis of the resulting LMIs.
\begin{proposition}\label{propo1}
$\kappa_{N}=O(N^{-\frac{1}{M}})$, $N\to \infty$.
\end{proposition}
\begin{proof}
According to the correspondence between eigenvalues and indices given in \eqref{eigenvalue1} and \eqref{eigenvalue2}, we associate each $n \in \mathbb{N}$ with a pair $(\bar{n}_M, n_M)$, where $\bar{n}M = (n_1, \dots, n{M-1})$. Then we have
\begin{equation*}
\begin{array}{ll}
\phi_n(x_1,\dots,x_{M-1},0) = \Phi_{\bar{n}_M}(x_1,\dots,x_{M-1}) \cdot \sqrt{\frac{2}{a_M}}, \\
\Phi_{\bar{n}_M}(x_1,\dots,x_{M-1}) = \prod_{k=1}^{M-1} \sqrt{\frac{2}{a_k}} \sin\left(\frac{n_k\pi x_k}{a_k}\right).
\end{array}
\end{equation*}
The system $\{\Phi_{\bar{n}_M}\}_{\bar{n}_M \in \mathbb{N}^{M-1}}$ forms a complete orthonormal system in $L^2(\Gamma_1)$. 
For $c_m \in L^2(\Gamma_1)$, define $\tilde{c}_{m,\bar{n}_M} = \langle c_m, \Phi_{\bar{n}_M} \rangle_{\Gamma_1}$. Then ${\bf c}_{m,n} = \langle c_m, \phi_n \rangle_{\Gamma_1} = \sqrt{\frac{2}{a_M}} \tilde{c}_{m,\bar{n}_M}$, 
and thus ${\bf c}_{m,n}^2 = \frac{2}{a_M} \tilde{c}_{m,\bar{n}_M}^2$.
We write $\kappa_N$ as 
\begin{equation}\label{kappaNN}
\begin{array}{ll}
\kappa_N = \frac{2}{a_M} \sum_{m=1}^{d} \sum_{n=N+1}^{\infty} \frac{\tilde{c}_{m,\bar{n}_M}^2}{\lambda_n}.
\end{array}
\end{equation}
For $\lambda_{N+1}$, define the index $n^* = (n_1^*, \dots, n_M^*)$ where:
\begin{equation*}
\begin{array}{ll}
     n_j^* = \left\lfloor \frac{a_j}{\pi} \lambda_{N+1}^{1/2} \right\rfloor \quad \text{for } j=1,\dots,M-1,\\
n_M^* = \left\lfloor \frac{a_M}{\pi} \lambda_{N+1}^{1/2} + \frac12 \right\rfloor,
\end{array}
\end{equation*}
where $\lfloor a \rfloor$ denotes the integer part of $0<a\in\mathbb{R}$. 
By the eigenvalue asymptotics \eqref{george}, we have $n_j^* = O(N^{1/M})$, $N\to\infty$, for all $j=1,\dots,M$. 
We split the sum  $\sum_{n=N+1}^{\infty} \frac{\tilde{c}_{m,\bar{n}_M}^2}{\lambda_n}$ in  \eqref{kappaNN} into two parts:
\begin{equation}\label{sumsum}
\begin{array}{ll}
\sum_{n=N+1}^{\infty} \frac{\tilde{c}_{m,\bar{n}_M}^2}{\lambda_n} 
= S_1+S_2,\\
S_1= \sum_{n_M=1}^{n_M^*} \sum_{\substack{\bar{n}_M \in \mathbb{N}^{M-1} \\ |\bar{n}_M| \geq |\bar{n}_M^*|}} \frac{\tilde{c}_{m,\bar{n}_M}^2}{\lambda_n},\\
S_2= \sum_{\bar{n}_M \in \mathbb{N}^{M-1}} \sum_{n_M=n_M^*+1}^{\infty} \frac{\tilde{c}_{m,\bar{n}_M}^2}{\lambda_n},
\end{array}
\end{equation}
where $|\bar{n}_M^*|^2 = \sum_{j=1}^{M-1} (n_j^*)^2$.

For $S_1$, noting that $\sum_{\bar{n}_M} \tilde{c}_{m,\bar{n}_M}^2 \leq \|c_m\|_{L^2(\Gamma_1)}^2$ and $\lambda_n\geq \lambda_{N+1}$ for $n>N$, we have 
\begin{equation}\label{S1S1}
\begin{array}{ll}
    S_1 \leq \frac{1}{\lambda_{N+1}} \sum_{n_M=1}^{n_M^*} \sum_{\substack{\bar{n}_M \in \mathbb{N}^{M-1} \\ |\bar{n}_M| \geq |\bar{n}_M^*|}} \tilde{c}_{m,\bar{n}_M}^2 \\
\leq \frac{n_M^*}{\lambda_{N+1}} \|c_m\|_{L^2(\Gamma_1)}^2.
\end{array}
\end{equation}
For $S_2$, when $n_M \geq n_M^* + 1$, we use the bound $\lambda_n = \pi^2[ \sum_{j=1}^{M-1} ( {n_j}/{a_j} )^2 + ( \frac{n_M - \frac12}/{a_M} )^2 ] 
\geq {\pi^2}/{a_M^2} ( n_M - \frac12 )^2$. 
Therefore, 
\begin{equation}\label{S2S2}
\begin{array}{ll}
     S_2 \leq \sum_{\bar{n}_M \in \mathbb{N}^{M-1}} \sum_{n_M=n_M^*+1}^{\infty} \frac{\tilde{c}_{m,\bar{n}_M}^2}{\frac{\pi^2}{a_M^2} (n_M - \frac12)^2} \\
= \frac{a_M^2}{\pi^2} \|c_m\|_{L^2(\Gamma_1)}^2 \sum_{n_M=n_M^*+1}^{\infty} \frac{1}{(n_M - \frac12)^2} \\
\leq  \frac{a_M^2}{\pi^2(n_M^* - \frac12)}\|c_m\|_{L^2(\Gamma_1)}^2, 
\end{array}
\end{equation}
where the last inequality follows from  $\sum_{n_M=n_M^*+1}^{\infty} \frac{1}{(n_M - \frac12)^2} 
\leq \int_{n_M^*}^{\infty} \frac{dx}{(x - \frac12)^2} 
= \frac{1}{n_M^* - \frac12}$. 
Since $n_{M}^{*}=O(N^{1/M})$ and $\lambda_{N+1}=O(N^{2/M})$, $N\to\infty$, we get $S_1 = O(N^{-1/M})$ and $S_2  = O(N^{-1/M})$,  $N\to\infty$. 
Substituting \eqref{sumsum}, \eqref{S1S1} and \eqref{S2S2} into  \eqref{kappaNN}, we get the desired conclusion.
\end{proof}

From \eqref{wnwnrd}, \eqref{observer22300rd}, and \eqref{tail100rd}, we have  
\begin{equation}\label{error00rd}
  \begin{array}{ll}
 \mathrm{d} e_{n}(t)=[(-\lambda_{n}+q)e_{n}(t)  -l_{n} ({\bf C}e^{N}(t) + \zeta(t)) \\
 + \bar{f}_n(t)] \mathrm{d}t+ g_n(t)\mathrm{d}\mathcal{W}(t), ~1\leq n \leq N,  
  \end{array}
\end{equation}
where $\bar{f}_n= f_n-\hat{f}_n $.
Let $L=[L_0^{\mathrm{T}}, 0_{d\times (N-N_0)}]^{\mathrm{T}}$, ${\bf I}_{1}=[0_{N \times d}, I_N]^{\mathrm{T}}$, $\hat{f}^{N}=\mathrm{col}\{\hat{f}_n\}_{n=1}^{N}$, $\bar{f}^{N}=\mathrm{col}\{\bar{f}_n\}_{n=1}^{N}$, $ g^{N}=[g_1^{\mathrm{T}}, \dots g_N^{\mathrm{T}} ]^{\mathrm{T}} \in \mathbb{R}^{N\times r}$. 
By \eqref{eq19neumannrd}, \eqref{eq20neumannrdtail}, \eqref{observer22300rd}, \eqref{controllawrd}, and \eqref{error00rd}, we obtain 
 the closed-loop system:  
\begin{subequations}\label{eqXX00aa}
\begin{align}
&\mathrm{d}\hat{w}^{N}(t)=[(\tilde{A}- \tilde{\bf B}K)\hat{w}^{N}(t) + {\bf I}_{1}\hat{f}^{N}(t) ]\mathrm{d}t \nonumber \\
&~~~~~~+  {\bf I}_{1} L[{\bf C}e^{N}(t)+\zeta(t) ]\mathrm{d}t,  \label{eqXX0000aa}\\
&\mathrm{d}e^{N}(t)= [(A - L{\bf C})e^{N}(t) +\bar{f}^{N}(t) - L \zeta(t)]\mathrm{d}t  \nonumber \\
&~~~~~~+ g^{N}(t)\mathrm{d}\mathcal{W}(t),   \label{eqXX0011aa}\\
&\mathrm{d}X_{\mathrm{tail}}(\cdot,t)=[-\mathcal{A}_{\mathrm{tail}}X_{\mathrm{tail}}(\cdot,t)+ qX_{\mathrm{tail}}(\cdot,t)  \nonumber\\
&	~~~~~ +(\text{Id}-\mathbb{P}_{N}) \psi^{\mathrm{T}}(\cdot)K\hat{w}^{N}(t)] \mathrm{d}t \nonumber \\
&	~~~~~ + (\text{Id}-\mathbb{P}_{N})f (\cdot, w(t) + \psi^{\mathrm{T}}(\cdot)  {\bf u}(t))\mathrm{d}t  \label{eqXX0022aa}\\
&   ~~~~~ +(\text{Id}-\mathbb{P}_{N})g (\cdot, w(t) + \psi^{\mathrm{T}}(\cdot)  {\bf u}(t)) \mathrm{d}\mathcal{W}(t).  \nonumber
\end{align}
\end{subequations}
For the stability of  \eqref{eqXX00aa}, we consider the Lyapunov function \vspace{-0.25cm}
\begin{equation}\label{FSVdy}
\begin{array}{ll}
	V(t)=V_{P_1}(t)+ V_{P_{2}}(t) +V_{\mathrm{tail}}(t),\\
	V_{P_1}(t)=|\hat{w}^{N}(t)|^{2}_{P_1}, ~V_{P_{2}}(t)=|e^{N}(t)|^{2}_{P_{2}}, \\
	V_{\mathrm{tail}}(t)= \|X_{\mathrm{tail}}(\cdot,t)\|^2_{L^2},  
\end{array}
\end{equation}
where $0<P_1\in \mathbb{R}^{(N+d)\times (N+d)}$, $0<P_2\in \mathbb{R}^{N\times N}$. 
By Parseval's identity, we have \vspace{-0.2cm}
\begin{equation}\label{triangleineq}
\begin{array}{ll}
 \alpha_{\min} [\|w(\cdot,t)\|^{2}_{L^2}+|{\bf u}(t)|^2 + |e^N(t)|^2 ] \\
 \leq  V(t)
  \leq \alpha_{\max} [\|w(\cdot,t)\|^{2}_{L^2}+|{\bf u}(t)|^2+ |e^N(t)|^2], \vspace{-0.2cm}
\end{array}
\end{equation}
for some $0<\alpha_{\min}\leq \alpha_{\max}$.

For {$V_{\mathrm{tail}}$}, calculating the generator {$\mathcal{L}$} along \eqref{eqXX0022aa} (see \cite[P. 228]{chow2007stochastic}) we obtain
\begin{equation}\label{directeq6aabb}
 \begin{array}{ll}
		\mathcal{L}V_{\mathrm{tail}}(t) 
= 2  \langle X_{\mathrm{tail}}(\cdot,t), -\mathcal{A}_{\mathrm{tail}}X_{\mathrm{tail}}(\cdot,t)+ qX_{\mathrm{tail}}(\cdot,t) \rangle \\
		~~ + 2  \langle X_{\mathrm{tail}}(\cdot,t) ,  (\text{Id}-\mathbb{P}_{N}) \psi^{\mathrm{T}}(\cdot)K\hat{w}^{N}(t) \rangle \\
		
		~~ + 2  \langle X_{\mathrm{tail}}(\cdot,t) ,  (\text{Id}-\mathbb{P}_{N})f (\cdot, w(t) + \psi^{\mathrm{T}}(\cdot)  {\bf u}(t)) \rangle \\
		
~~+  \int_{\mathcal{O}}|(\text{Id}-\mathbb{P}_{N})g (x,z(x,t)) |^2 \mathrm{d}x \\
	= \sum_{n=N+1}^{\infty} [2 (-\lambda_{n}+q)w_{n}^{2}(t) - 2  w_{n}(t) {\bf b}_{n}K\hat{w}^{N}(t)] \\
	~~	 + \sum_{n=N+1}^{\infty}2 w_{n}(t)[\hat{f}_{n}(t)+\bar{f}_n(t)]  +  \sum_{n=N+1}^{\infty}g_n^2(t).
	\end{array} 
\end{equation}
By Young's inequality, we have for $\alpha>0$:
\begin{equation}\label{dyeq4747rd00}
	\begin{array}{ll}
		-\sum_{n=N+1}^{\infty}2  w_{n}{\bf b}_{n}K\hat{w}^{N} \\
		\leq  \alpha  \sum_{n=N+1}^{\infty}  w^{2}_{n} + \frac{1}{\alpha}  \sum_{n=N+1}^{\infty} |{\bf b}_{n}|^2|K\hat{w}^{N}|^{2}\\
		= \alpha  \sum_{n=N+1}^{\infty}  w_{n}^{2}+ \frac{\|\psi\|^2_{N}}{\alpha} |K\hat{w}^{N}|^{2},
	\end{array}
\end{equation}
where $\|\psi\|^2_{N} = \sum_{n=N+1}^{\infty} |{\bf b}_{n}|^2$. 
For $V_{P_1}$ and $V_{P_2}$, calculating the generator $\mathcal{L}$  along \eqref{eqXX0000aa} and \eqref{eqXX0011aa} (see \cite[P. 149]{klebaner2012introduction}), we get \vspace{-0.25cm}
\begin{equation}\label{dotVobrd1}
 \begin{array}{ll}
		\mathcal{L} V_{P_1}+2\delta V_{P_1} = (\hat{w}^{N})^{\mathrm{T}} [P_{1}(\tilde{A}- \tilde{\bf B}K) 
		 +(\tilde{A}- \tilde{\bf B}K)^{\mathrm{T}}P_1 \\
		~~~ +2\delta P_{1} ]\hat{w}^{N}  +2(\hat{w}^{N})^{\mathrm{T}}P_{1}{\bf I}_1 (\hat{f}^{N} + L[{\bf C}e^{N} +\zeta]),\\
		\mathcal{L} V_{P_2}+2\delta V_{P_2} = (e^{N})^{\mathrm{T}} [P_{2}(A - L{\bf C}) +(A - L{\bf C})^{\mathrm{T}}P_2 \\
		 ~~~ +2\delta P_{2} ]e^{N} + 2(e^{N})^{\mathrm{T}}P_{2}[\bar{f}^{N} - L\zeta] + (g^{N})^{\mathrm{T}}{P_2} g^{N}. \vspace{-0.1cm}
	\end{array} 
\end{equation}
From \eqref{sigma0}, we have
\begin{equation}\label{sigma0111aa}
\setlength{\arraycolsep}{0.01pt}			\begin{array}{ll}
 \sum_{n=1}^{\infty}\hat{f}_n^2(t)=\int_{0}^{1}| f(x,\hat{w}(x,t)+ \psi^{\mathrm{T}}(x){\bf u}(t) ) |^2\mathrm{d}x \\
 \leq \sigma_f^2\int_{0}^{1}| \hat{w}(x,t)+ \psi^{\mathrm{T}}(x){\bf u}(t) |^2\mathrm{d}x  = \sigma_f^2 |\hat{w}^N(t) |^2_{\Lambda_2}   , \\		
 \sum_{n=1}^{\infty}\bar{f}_n^2(t)
 \leq \sigma_f^2 \int_{0}^{1}|w(x,t) - \hat{w}(x,t) |^2\mathrm{d}x \\
 =\sigma_f^2 |e^{N}(t)|^2 + \sigma_f^2 \sum_{n=N+1}^{\infty}w_n^{2}(t) ,\\

 \sum_{n=1}^{\infty}g_n^2(t) \leq \sigma_g^2\int_{0}^{1}| w(x,t)+ \psi^{\mathrm{T}}(x){\bf u}(t) |^2\mathrm{d}x \\
  \leq 2\sigma_g^2 \left| \Lambda_3 \hat{w}^{N}(t)+ \Lambda_4 e^{N}(t)  \right|^2 +  2\sigma_g^2 \sum_{n=N+1}^{\infty}w_n^{2}(t),
		\end{array}
\end{equation}
where 
\begin{equation*}
	\begin{array}{ll}
	\Lambda_2=  \left[ \begin{array}{ccc}
		\Psi & -{\bf B}^{\mathrm{T}} \\
		-{\bf B} & I_N
		\end{array} \right], ~~ \Psi= ( \int_{\mathcal{O}}\psi_i(x)\psi_j(x)\mathrm{d}x )_{d\times d},\\
		\Lambda_3=  \left[ \begin{array}{ccc}
		\Psi^{\frac{1}{2}} & 0  \\
		0 & I_N
		\end{array} \right], ~~\Lambda_4=  \left[ \begin{array}{ccc}
		  0_{d\times N} \\
		 I_N
		\end{array} \right].
	\end{array}
\end{equation*}


Let $\eta=\mathrm{col}\{\hat{w}^{N},e^{N},\zeta, \hat{f}^N, \bar{f}^{N}\}$.
By \eqref{directeq6aabb}, \eqref{dyeq4747rd00}, \eqref{dotVobrd1},  and using S-Procedure and \eqref{eq3333rd}, \eqref{sigma0111aa},  we obtain for $\beta_i>0$, $i=1,2,3,4$,
\begin{equation}\label{dotVob000}
	\begin{array}{ll}
		\mathcal{L}V(t)+2\delta V(t) +\beta_1 [\kappa_{N} \sum_{n=N+1}^{\infty}\lambda_{n} w_{n}^{2}(t) -|\zeta(t)|^{2} ]\\
		~~~ +\beta_2 [\sigma_f^2 |e^{N}(t)|^2 + \sigma_f^2 \sum_{n=N+1}^{\infty}w_n^{2}(t)- \sum_{n=1}^{\infty}\bar{f}_n^2(t)] \\
         ~~~+\beta_3 [ \sigma_f^2 | \hat{w}^N(t) |^2_{\Lambda_2}  -\sum_{n=1}^{\infty}\hat{f}_n^2(t)]\\
		 ~~~+\beta_4 [ 2\sigma_g^2 | \Lambda_3\hat{w}^{N}(t) +  \Lambda_4 e^{N}(t)|^2  \\
~~~+  2\sigma_g^2 \sum_{n=N+1}^{\infty}w_n^{2}(t)  - \sum_{n=1}^{\infty} g_n^2(t)]\\
	\leq  \eta^{\mathrm{T}}(t) \Theta	\eta(t)   +(g^{N}(t))^{\mathrm{T}}(P_2 -\beta_4 I) g^{N}(t) \\
~~~ +\sum_{n=N+1}^{\infty}   g_n^{\mathrm{T}}(t)(I-\beta_4 I)g_n(t)   \\
~~~ +  \sum_{n=N+1}^{\infty} {\small \left[ \begin{array}{cc}
		w_n (t) \\
		\hat{f}_n (t)\\
\bar{f}_n (t)
		\end{array} \right]^{\mathrm{T}} \Gamma_n \left[ \begin{array}{cc}
		w_n (t) \\
		\hat{f}_n (t)\\
\bar{f}_n(t)
		\end{array} \right] }<0,  
	\end{array}
\end{equation}
provided  
\begin{subequations}\label{eq50}
 \begin{align}
&\Gamma_n := \left[ \begin{array}{ccc}
		\Upsilon_n & 1 & 1 \\
		*& -\beta_3 & 0 \\
*&*& -\beta_2
		\end{array} \right] <0 ,~~n>N, \label{eq50aa}\\
& P_2 -\beta_4 I<0, ~~~~~~1-\beta_4 <0,  \label{eq50bb} \\
&\Theta: ={ \setlength{\arraycolsep}{3pt} \left[ \begin{array}{cccccc}
		\Theta_{11} &P_{1}{\bf I}_1{L}{\bf C}  & P_{1}{\bf I}_1L  & P_{1}{\bf I}_1 & 0 \\
		*& \Theta_{22} &-P_2 L &  0 & P_2 \\
		*& * &-\beta_1 I & 0& 0 \\
		*&*&*& -\beta_3 I & 0\\
		*&*&*&*& -\beta_{2}I
		\end{array} \right] }      \nonumber \\
&+ 2\beta_4 \sigma_g^2 [\Lambda_3, \Lambda_4, 0, 0, 0]^{\mathrm{T}}[\Lambda_3, \Lambda_4, 0, 0, 0]<0, \label{eq50cc}
\end{align}
\end{subequations}
where 
\begin{equation*}
	\begin{array}{ll}
		\Upsilon_n = -2(\lambda_{n}-q-\delta)  + \alpha   +  \beta_1\kappa_{N}\lambda_{n}\\
		~~~~~~~ +\beta_2\sigma_f^2+2\beta_4\sigma_g^2 ,\\
		\Theta_{11}=P_{1}(\tilde{A}- \tilde{\bf B}K)+(\tilde{A}- \tilde{\bf B}K)^{\mathrm{T}}P_1  \\
~~~~~~~ +2\delta P_{1}  +  \frac{1}{\alpha}\|\psi\|^{2}_{N} K^{\mathrm{T}}K + \beta_3\sigma_f^2 \Lambda_2  , \\
\Theta_{22}= P_{2}(A- L{\bf C})+(A- L{\bf C})^{\mathrm{T}}P_2 \\
 ~~~~~~~ + 2\delta P_2  + \beta_2\sigma_f^2I.   
	\end{array} 
\end{equation*}
By employing It\^{o}'s formula for $\mathrm{e}^{2\delta t}V_{P_1}(t)$, $\mathrm{e}^{2\delta t}V_{P_2}(t)$, along stochastic ODEs \eqref{eqXX0000aa}, \eqref{eqXX0011aa} (see \cite[Theorem 4.18]{klebaner2012introduction}), and taking expectation on both sides, we have  
\begin{equation}\label{ito1}
	\begin{array}{ll}
		 \mathbb{E} [\mathrm{e}^{2\delta t}V_{P_k}(t)]= \int_{0}^{t} \mathrm{e}^{2\delta s}\mathbb{E}[\mathcal{L}V_{P_k}(s)+ 2\delta V_{P_k}(s)]\mathrm{d}s \\
~~~~~~~~~~~~~~~~~~~~~ +\mathrm{e}^{2\delta t}V_{P_k}(0),~~ k=1,2.  
	\end{array}
\end{equation}
Employing It\^{o}'s formula for $\mathrm{e}^{2\delta t}V_{\mathrm{tail}}(t)$ along \eqref{eqXX0022aa} (see \cite[Theorem 7.2.1]{chow2007stochastic}) and taking expectation on both sides, we arrive at  
\begin{equation}\label{ito2}
	\begin{array}{ll}
		\mathbb{E} [\mathrm{e}^{2\delta t}V_{\mathrm{tail}}(t)]= V_{\mathrm{tail}}(0) \\
+\int_{0}^{t} \mathrm{e}^{2\delta s}\mathbb{E}[\mathcal{L}V_{\mathrm{tail}}(s)+ 2\delta V_{\mathrm{tail}}(s)]\mathrm{d}s.  
	\end{array}
\end{equation}
Using \eqref{ito1} and \eqref{ito2}, we obtain
\begin{equation}\label{eq49}
	\begin{array}{ll}
		\mathrm{e}^{2\delta t}\mathbb{E}V(t)=V(0) +  \int_{0}^{t} \mathrm{e}^{2\delta s}\mathbb{E}[\mathcal{L}V(s)+ 2\delta V(s)  ]\mathrm{d}s \\
		\overset{\eqref{dotVob000}}{\leq} V(0),~t\geq 0,  
	\end{array}
\end{equation}
which together with \eqref{triangleineq} implies
\begin{equation}\label{stability}
	\begin{array}{ll}
		\mathbb{E}[|{\bf u}(t)|^{2} + \|w(\cdot,t)\|_{L^{2}}^{2}+ |e^N(t)|^2]\\
		\leq D_0\mathrm{e}^{-2\delta t}\|z_{0}\|_{L^{2}}^{2},~~t\geq 0,
	\end{array}	
	\end{equation}
for some $D_0\geq 1$.

For controller design, we introduce $Q_1=P_{1}^{-1}$, $Q_2=P_{2}^{-1}$, $Y=KQ_{1}$, $\tilde{\beta}_i=\frac{1}{\beta_i}$, $i=2,3,4$.
Multiplying $\Theta$ in \eqref{eq50cc}  by  $\mathrm{diag}\{Q_{1}, Q_{2}, I, \tilde{\beta}_3I, \tilde{\beta}_2I\}$ from the right and the left, and using Schur complement, we find that \eqref{eq50cc}  holds iff
\begin{equation}\label{eq4343}
		{\small \left[ \begin{array}{cccccccccc}
		\bar{\Theta}_{11} & \bar{\Theta}_{12} & {\bf I}_1{L} & \tilde{\beta}_3{\bf I}_1 & 0 & \bar{\Theta}_{16}  & \bar{\Theta}_{17}  & 0 \\
		* & \bar{\Theta}_{22}  & -L & 0  & \tilde{\beta}_2 I & \bar{\Theta}_{26}  &0& \sigma_f Q_2 \\
		*& *& -\beta_1 I & 0 & 0& 0 &0& 0 \\
		*&*&*& -\tilde{\beta}_3 I & 0& 0 &0& 0\\
		 *&*&*& *& -\tilde{\beta}_2 I & 0 &0& 0\\
		 *&*&*& *& *& -\frac{\tilde{\beta}_4}{2} &0& 0 \\
		 *&*&*& *& *& *& \bar{\Theta}_{77} & 0 \\
		  *&*&*& *& *& *& *& -\tilde{\beta}_2 I
		\end{array} \right]<0, }
\end{equation}
where
\begin{equation*}
	 \begin{array}{ll}
\bar{\Theta}_{11}=\tilde{A}^{\mathrm{T}}Q_1 + Q_1 \tilde{A}- \tilde{\bf B}Y-Y^{\mathrm{T}}\tilde{\bf B}^{\mathrm{T}} +2\delta Q_1, \\
\bar{\Theta}_{12} ={\bf I}_1{L}{\bf C}Q_2, ~ \bar{\Theta}_{16}=\sigma_g Q_1 \Lambda_3^{\mathrm{T}} , ~\bar{\Theta}_{26}=\sigma_g Q_2 \Lambda_4^{\mathrm{T}} , \\
\bar{\Theta}_{22}= (A - L{\bf C})Q_2+Q_2(A - L{\bf C})^{\mathrm{T}} +2\delta Q_{2},\\
\bar{\Theta}_{17} =[\sigma_f Q_1\Lambda_2^{\frac{1}{2}}, Y^{\mathrm{T}} ], ~\bar{\Theta}_{77}=-\mathrm{diag}\{\tilde{\beta}_3 I, \alpha\|\psi\|^{-2}_{N} I \}.
	\end{array}
\end{equation*}
By using Schur complement, \eqref{eq50aa} and \eqref{eq50bb} hold iff
\begin{eqnarray}\label{eq44}
	\begin{array}{ll}
		\tilde{\beta}_4 I -Q_2 <0, ~~\tilde{\beta}_4-1<0, \\
		 \left[ \begin{array}{cccccc}
		\tilde{\Upsilon} & \tilde{\beta}_3 & \tilde{\beta}_2 & \sigma_f & \sigma_{g} \\
		*& -\tilde{\beta}_3 & 0 & 0 & 0 \\
*&*& -\tilde{\beta}_2 & 0& 0 \\
*&*&*& -\tilde{\beta}_2 & 0 \\
*&*&*&*& -\frac{\tilde{\beta}_4}{2}
		\end{array} \right] <0,\\
		\tilde{\Upsilon} =-2\lambda_{N+1}+2q+2\delta  + \alpha   +  \beta_1\kappa_{N}\lambda_{N+1}.
	\end{array}
\end{eqnarray}
If LMIs \eqref{eq4343} and \eqref{eq44} are feasible, the controller gain is obtained by $K=YQ_{1}^{-1}$.
Summarizing, we have:
 \begin{theorem}\label{thm1}
	Let $N_0\in\mathbb{N}$ satisfy \eqref{NNNed}, $N\geq N_0$, and $d$ be the maximum geometric multiplicity of $\{\lambda_{n}\}_{n=1}^{N_0}$.
	 Consider \eqref{eq19neumannrd}, \eqref{eq20neumannrd} with $f,g$ satisfying \eqref{sigma0}, initial value $z_0\in\mathcal{D}(\mathcal{A})$, control law \eqref{controllawrd}, boundary measurement \eqref{measurementnewrd}, and observer \eqref{observer22300rd}. 
	Choose the shape functions $\{b_{j}\}_{j=1}^{N_0}$, $\{c_{j}\}_{j=1}^{N_0}$ such that \eqref{eq20dy} holds. Let $L_0$ be obtained from \eqref{observergainfulldy}.  Let there exist scalars $\alpha>0, \beta_1>0$, $\tilde{\beta}_i>0$, $i=2,3,4$, matrices $0<Q_1\in \mathbb{R}^{(N+d)\times (N+d)}$, $0<Q_2\in \mathbb{R}^{N\times N}$ and $Y\in \mathbb{R}^{d \times (N+d)}$ such that LMIs  \eqref{eq4343} and \eqref{eq44} are feasible.
	Then the above system with controller gain $K=YQ_{1}^{-1}$ satisfies \eqref{stability} for some $D_0>1$.
\end{theorem}

\begin{remark}\label{remark2}
In Theorem \ref{thm1}, we provide LMIs for finding the observer dimension $N$ and the upper bounds on the Lipschitz constants $\sigma_f, \sigma_g$. We claim that
the LMIs \eqref{eq4343} and \eqref{eq44} (equivalently, matrix inequalities \eqref{eq50}) are always feasible for large enough $N$.
To show the feasibility,
we choose $\beta_2=\beta_3=\frac{1}{\sigma_f}$, $\beta_4=\frac{1}{\sigma_g}$, and let $\sigma_f,\sigma_g\to 0^+$. We find that inequalities \eqref{eq50} hold if  
\begin{subequations}\label{eq44abcd}
\setlength{\arraycolsep}{3pt} \begin{align}
	& -2\lambda_{N+1}+2q+2\delta  + \alpha   +  \beta_1\kappa_{N}\lambda_{N+1}<0,  \label{eq44abcdaa}\\
	&{\small \left[ \begin{array}{cccccc}
		\Theta_{11}+ \frac{1}{\alpha}\|\psi\|^{2}_{N} K^{\mathrm{T}}K & P_{1}{\bf I}_1L{\bf C}  & P_{1}{\bf I}_1L   \\
		*& \Theta_{22} &-P_2 L  \\
		*& * &-\beta_1 I
		\end{array} \right]} <0,  \label{eq44abcdbb}
\end{align}
\end{subequations}
where $\Theta_{11}= P_{1}(\tilde{A}- \tilde{\bf B}K)+(\tilde{A}- \tilde{\bf B}K)^{\mathrm{T}}P_1+2\delta P_{1} $, 
$\Theta_{22}= P_{2}(A- L{\bf C})+(A- L{\bf C})^{\mathrm{T}}P_2 + 2\delta P_2$. 
Let $K=[K_0, 0_{d\times (N-N_0)}]$ with $K_0\in \mathbb{R}^{d\times N_0}$, $P_1=\mathrm{diag}\{\hat{P}_{1},p_{1}I_{N-N_0}\}$ and $P_2=\mathrm{diag}\{\hat{P}_{2},p_{2}I_{N-N_0}\}$  with scalars $p_1, p_2>0$ and $ 0<\hat{P}_{1}\in\mathbb{R}^{(N_{0}+d)\times (N_{0}+d)}$, $0<\hat{P}_{2}\in\mathbb{R}^{N_{0}\times N_{0}}$.
Substituting such $P_1$ and $P_2$ into \eqref{eq44abcd} and letting $p_{1}\to 0^{+}$, $p_{2}\to \infty$, since $A_1+\delta I<0$ (from \eqref{NNNed}), we find that \eqref{eq44abcdbb} holds if  
\begin{equation}\label{505050}
\setlength{\arraycolsep}{2pt}	  \begin{array}{ll}
		\left[\begin{array}{cccc}	
		 R_{1} + \| \psi\|^{2}_{N} K_{0}^{\mathrm{T}}K_{0}  & \hat{P}_{1}{\bf I}_0 L_{0}{\bf C}_{0} & \hat{P}_{1}{\bf I}_0 L_{0}  \\
	     *& R_{2} &  - \hat{P}_{2}L_{0}L_{0} \\
*&*& -\beta_1 I
	\end{array}  \right]  <0,\\
	R_{1}=\hat{P}_{1}(\tilde{A}_{0}- \tilde{\bf B}_{0}K_{0})+(\tilde{A}_{0}- \tilde{\bf B}_{0}K_{0})^{\mathrm{T}}\hat{P}_{1} +2\delta \hat{P}_{1} ,  \\
	R_{2}=\hat{P}_{2}(A_{0}-L_{0}{\bf C}_{0})+(A_{0}-L_{0}{\bf C}_{0})^{\mathrm{T}}\hat{P}_{2}+2\delta \hat{P}_{2},  
	\end{array}
\end{equation}
where ${\bf I}_0 = [0_{d \times N_0}, I_{N_0}]^{\mathrm{T}}$.
Multiplying \eqref{505050} by {\small$\mathrm{diag}\{\hat{P}_{1}^{-1},I\}$} from left and
right, and introducing {\small$S=\hat{P}_{1}^{-1}$}, we have that \eqref{505050} holds iff  
 \begin{equation}\label{505050aa}
\setlength{\arraycolsep}{3pt}	 \begin{array}{ll}
		 \left[\begin{array}{cccc}	
		 \hat{R}_{1} + \| \psi\|^{2}_{N} SK_{0}^{\mathrm{T}}K_{0}S  & {\bf I}_0 L_{0}{\bf C}_{0} & {\bf I}_0 L_{0}  \\
	     *& R_{2} &  - \hat{P}_{2}L_{0} \\
*&*& -\beta_1 I
	\end{array}  \right]   <0,\\
	\hat{R}_{1}=(\tilde{A}_{0}- \tilde{\bf B}_{0}K_{0})S+S(\tilde{A}_{0}- \tilde{\bf B}_{0}K_{0})^{\mathrm{T}} +2\delta S ,  \\
	R_{2}=\hat{P}_{2}(A_{0}-L_{0}{\bf C}_{0})+(A_{0}-L_{0}{\bf C}_{0})^{\mathrm{T}}\hat{P}_{2}+2\delta \hat{P}_{2}.  
	\end{array}
\end{equation}
 Let $S$ and $\hat{P}_{2}$ solve the Lyapunov equations:  
 \begin{equation*}
	\begin{array}{ll}
 		S(\tilde{A}_{0}- \tilde{\bf B}_{0}K_{0}+\delta I)^{\mathrm{T}}+(\tilde{A}_{0}- \tilde{\bf B}_{0}K_{0}+\delta I)S= -\chi_{1}I,\\
 		\hat{P}_{2}(A_{0}+L_{0}{\bf C}_{0}+\delta I)+(A_{0}+L_{0}{\bf C}_{0}+\delta I)^{\mathrm{T}}\hat{P}_{2}=- \chi_{2}I,  
 	\end{array} 
 \end{equation*}
where $\chi_1,\chi_2>0$ are independent of $N$ and satisfy $-\chi_{1}I+\chi_{2}^{-1}{\bf I}_0 L_{0}{\bf C}_{0}{\bf C}_{0}^{\mathrm{T}}L_{0}^{\mathrm{T}}{\bf I}_0^{\mathrm{T}}<0$. We have $\|S\|=O(1), \|\hat{P}_{2}\|=O(1)$, $N\to \infty$. Substituting above $S$ and $\hat{P}_{2}$ into \eqref{505050aa}, choosing $\beta_1=N^{\frac{1}{2M}}$, and using Proposition \ref{propo1}, we have that \eqref{eq44abcdaa} and \eqref{505050aa} are always feasible for large enough $N$. By continuity, LMIs \eqref{eq50} hold for small enough $\sigma>0$ and large enough $N>0$.
\end{remark}

\section{Noise-to-state stability for additive noise}
In this section, we consider the following system with persistent noise: 
\begin{equation}\label{eq14neumannrdNSS}
\setlength{\arraycolsep}{1pt}		\begin{array}{ll}
		\mathrm{d} z(x,t)=[\Delta z(x,t)+qz(x,t)+f(x,z(x,t))]\mathrm{d}t \\
		~~~~~~~~~~~~~ +g(x)\Sigma(t)\mathrm{d}\mathcal{W}(t), ~~ t\geq 0, ~x\in \mathcal{O}, \\
\frac{\partial z(x,t)}{\partial \nu}=u(x,t), x\in\Gamma_{1}, ~   z(x,t)=0, ~x\in\Gamma_{2}, \\
		z(x,0)=z_0(x),
	\end{array}
\end{equation}
with measurement \eqref{measurement}, where $g: \mathcal{O} \to \mathbb{R}^r$, $\Sigma(t): \mathbb{R}^r \to \mathbb{R}^r$ is treated as unknown, and $\mathcal{W}(t)$ is a standard $r$-dimensional Brownian motion defined on $(\Omega, \mathcal{F}, \mathbb{P})$. 
By considering Lyapunov function \eqref{FSVdy} and using arguments similar to Sections \ref{sec2.2}-\ref{sec2.4}, we obtain
\begin{equation}\label{dotVob111}
	\begin{array}{rr}
		\mathcal{L}V(t)+2\delta V(t) 
		 \leq \mathrm{Trace}\{(g^{N}(t))^{\mathrm{T}}P_2 g^{N}(t)\} \\
 ~~+\sum_{n=N+1}^{\infty} \mathrm{Trace}\{ g_n^{\mathrm{T}}(t)g_n(t) \},
	\end{array}
\end{equation}
provided the following LMIs are feasible:
\begin{equation}\label{noisetostate1}
	 \begin{split}
&\left[ \begin{array}{ccc}
		\hat{\Upsilon} & \tilde{\beta}_3 & 1 \\
		*& -\tilde{\beta}_3 & 0 \\
*&*& -\beta_2
		\end{array} \right]<0, \\	
&\left[ \begin{array}{cccccccccc}
		\bar{\Theta}_{11} &{\bf I}_1{L}{\bf C}  & {\bf I}_1{L} & {\bf I}_1\tilde{\beta}_3 & 0 & \bar{\Theta}_{17} \\
		* & \hat{\Theta}_{22}  & -P_2 L & 0  & -P_2 & 0  \\
		*& *& -\beta_1 I & 0 & 0& 0  \\
		*&*&*& -\tilde{\beta}_3 I & 0& 0 \\
		 *&*&*& *& -\beta_2 I & 0 \\
		 *&*&*& *& *& \bar{\Theta}_{77}
		\end{array} \right]<0, 
	\end{split}
\end{equation}
where $\bar{\Theta}_{11}, \bar{\Theta}_{17}, \bar{\Theta}_{77}$ are defined below \eqref{eq4343}, and 
\begin{equation*}
	\begin{array}{ll}
	\hat{\Upsilon} =-2\lambda_{N+1}+2q+2\delta  + \alpha   +  \beta_1\kappa_{N}\lambda_{N+1}+\beta_2\sigma_f^2,\\	
	\hat{\Theta}_{22}= P_{2}(A- L{\bf C})+(A- L{\bf C})^{\mathrm{T}}P_2 + 2\delta P_2  + \beta_2\sigma_f^2I.	
	\end{array}
\end{equation*}
In LMIs \eqref{noisetostate1}, $0<\beta_1, \beta_2, \tilde{\beta}_3\in \mathbb{R}$, $0<Q_1\in \mathbb{R}^{(N+d)\times (N+d)}$, $0<P_2\in\mathbb{R}^{N\times N}$, $Y\in \mathbb{R}^{(N+d)\times d}$ are decision variables.
By arguments similar to Remark \ref{remark2}, it is easy to conclude that LMIs \eqref{noisetostate1} are always feasible for large enough $N$ and small enough $\sigma_f>0$. 
By employing Parseval's equality, we obtain from \eqref{dotVob111} that
\begin{equation}\label{eqeq56}
 \begin{array}{ll}
  \mathcal{L}V(t) +2\delta V(t) \\
  \leq \max\{\lambda_{\max}(P_2),1\}\sum_{n=1}^{\infty} \mathrm{Trace}\{\Sigma^{\mathrm{T}}(t) g^{\mathrm{T}}_{n}(t)g_{n} (t)\Sigma(t)\} \\
 \leq \max\{\lambda_{\max}(P_2),1\} \sum_{n=1}^{\infty} \mathrm{Trace}\{g^{\mathrm{T}}_{n}(t)g_{n} (t)\} |\Sigma(t)\Sigma^{\mathrm{T}}(t)| \\
= \max\{\lambda_{\max}(P_2),1\}\mathrm{Trace}\{\int_{\mathcal{O}}g^{\mathrm{T}}(x)g(x)\mathrm{d}x\}|\Sigma(t)\Sigma^{\mathrm{T}}(t)|  \\
= \rho |\Sigma(t)\Sigma^{\mathrm{T}}(t)| ,   
\end{array}
\end{equation}
where $\rho =\max\{\lambda_{\max}(P_2),1\}\mathrm{Trace}\{\int_{\mathcal{O}}g^{\mathrm{T}}(x)g(x)\mathrm{d}x\}$. 
 By employing It\^{o}'s formulas \eqref{ito1}-\eqref{eq49}, \eqref{eqeq56}, and arguments similar to Theorem 4.1 in \cite{deng2001stabilization}, we have 
\begin{equation}\label{NSS}
\begin{split}
& \mathbb{E}[\|w(\cdot,t)\|^2_{L^2} + |{\bf u}(t)|^2+ |e^N(t)|^2] \\
& \leq D_0\mathrm{e}^{-2\delta t}\|z_0\|^2_{L^2} +\frac{\rho}{2\delta}\sup_{0\leq s\leq t}|\Sigma(s)\Sigma^{\mathrm{T}}(s)|, ~t\geq 0,  
\end{split}
\end{equation}
for some $D_0\geq 1$. This shows that the closed-loop system is mean-square noise-to-state stable. 
Summarizing, we have:
 \begin{theorem}\label{thm2}
	Let $N_0\in\mathbb{N}$ satisfy \eqref{NNNed}, $N\geq N_0$, and $d$ be the maximum geometric multiplicity of $\{\lambda_{n}\}_{n=1}^{N_0}$.
	 Consider \eqref{eq19neumannrd}, \eqref{eq20neumannrd} with $f$ satisfying \eqref{sigma0}, $g(\cdot,z) = g(\cdot)\in L^2(\mathcal{O})$, initial value $z_0\in\mathcal{D}(\mathcal{A})$, control law \eqref{controllawrd}, and boundary measurement \eqref{measurementnewrd}. Choose the shape functions $\{b_{j}\}_{j=1}^{N_0}$, $\{c_{j}\}_{j=1}^{N_0}$ such that \eqref{eq20dy} holds. Let $L_0$ be obtained from \eqref{observergainfulldy} and 
  let there exist scalars $\alpha,\beta_1,\beta_2,\tilde{\beta}_3>0$, matrices $0<Q_1\in \mathbb{R}^{(N+d)\times (N+d)}$, $0<P_2\in \mathbb{R}^{N\times N}$ and $Y\in \mathbb{R}^{d \times (N+d)}$  such that LMIs  \eqref{noisetostate1}  are feasible.	Then the above system with controller gain $K=YQ_{1}^{-1}$ satisfies \eqref{NSS} for some $D_0\geq 1$. Moreover, the LMIs  \eqref{noisetostate1} are always feasible for large enough $N$ and small enough $\sigma_f>0$.
\end{theorem}

\begin{remark}\label{remark9}
For the linear deterministic case of system \eqref{eq14neumannrdNSS} with $M=2,3$ and $f=0, g=0$, under point in-domain measurement and Dirichlet actuation, \cite{lhachemi2023boundary} suggested a finite-dimensional output-feedback control. They suggested a scaling method to ensure the feasibility of the stability conditions as $N\to\infty$.   
We show below that this method cannot be extended to 
achieve noise-to-state stability in our present stochastic setting. Consider \eqref{eq14neumannrdNSS} with $f=0$.  Using the scaling technique from \cite{lhachemi2023boundary} yields the dynamics for $e^{N-N_0}=[e_{N_0+1},\dots,e_N]^{\mathrm{T}}$:
\begin{equation*}
    \begin{array}{ll}
        \tilde{e}^{N-N_0}(t)=\Lambda^{N-N_0}{e}^{N-N_0}(t), \\
        \mathrm{d}\tilde{e}^{N-N_0}(t)=A_1\tilde{e}^{N-N_0}(t) \mathrm{d}t + \Lambda^{N-N_0} g^{N-N_0}(t)\Sigma(t)\mathrm{d}\mathcal{W}(t),
    \end{array}
\end{equation*}
where $A_1=\mathrm{diag}\{-\lambda_{n}+q\}_{n=N_0+1}^{N}$,  $\Lambda^{N-N_0}=\mathrm{diag}\{\lambda_{n}-q\}_{n=N_0+1}^{N}$, $ g^{N-N_0}=[g_{N_0+1}^{\mathrm{T}}, \dots g_N^{\mathrm{T}} ]^{\mathrm{T}}$. Applying It\^{o}'s formula to $|\tilde{e}^{N-N_0}(t)|_{P_1}^2$ with $0<P_1\in\mathbb{R}^{(N-N_0)\times (N-N_0)}$ (noting that the $P_1$ involved in $P$ in (37) of \cite{lhachemi2023boundary} must satisfy $|P_1|=O(1)$, $N\to\infty$), produces a quadratic term $|\Lambda^{N-N_0}g^{N-N_0}(t)\Sigma(t)|_{P_1}^2$. The latter is not uniformly bounded and diverges as $N\to\infty$. This demonstrates that the scaling method in \cite{lhachemi2023boundary} cannot simultaneously guarantee the noise-to-state stability of system \eqref{eq14neumannrdNSS} and 
the feasibility of the sufficient conditions as  $N\to\infty$.
\end{remark}

\section{Numerical examples}\label{example}
In this section, we consider 2D and 3D cases, respectively to demonstrate the efficiency of our method. First, for the 2D case, i.e., $\mathcal{O}=[0,a_1]\times [0,a_2]$, we set $a_{1}= \sqrt{1.5}$, $a_{2}=1$, and $q=28.8$, which results in an unstable open-loop system with at least $3$ unstable modes ({$\lambda_{1}<\lambda_{2}=\lambda_{3}<q$}). 
Choosing $\delta=10^{-3}$, $N_{0}=3$, we have the maximum geometric multiplicity of $\{\lambda_{n}\}_{n=1}^{3}$ as $d=2$. 
Considering $c_{1}(x_1)=\frac{1}{\sqrt{a_1a_2}}\sin(\frac{\pi x_1}{a_1})$, $c_{2}(x_1)=\frac{0.1}{\sqrt{a_1a_2}}\sin(\frac{2\pi x_1}{a_1})$, $x_1\in[0,a_1]$, we have ${\bf C}_{0}= {\scriptsize \left[
               \begin{array}{cccc}
                1 &0 &1 \\
                0 &0.1 & 0
               \end{array}
             \right]}$, which satisfies \eqref{eq20dy}.
             The observer gain {$L_0$} is found from \eqref{observergainfulldy} 
and is given by 
\begin{equation*}
	L_0= \left[ \begin{array}{cccc}
		  55.0284   &      0 \\
         0 & 15.0583 \\
    -2.4175  &       0
		\end{array}  \right]. 
\end{equation*}
We consider the shape functions $b_1$, $b_2$ and the corresponding $\psi_1$, $\psi_2$ as follows:
\begin{equation}\label{shapeb}
 \begin{array}{ll}
b_{1}(x_1)=\frac{1}{\sqrt{a^*_1a^*_2}}\sin(\frac{\pi x_1}{a^*_1})\chi_{[0,a^*_1]}(x_1), ~a^*_1=0.8a_1, \\ 
b_{2}(x_1)=\frac{1.2}{\sqrt{a^*_1a_2}}\sin(\frac{2\pi x_1}{a^*_1})\chi_{[0,a^*_1]}(x_1), ~a^*_2=0.8a_2,\\
\psi_1(x)=-\frac{a^*_2}{\pi} b_{1}(x_1)\sin(\frac{\pi x_2}{a_2^*})\chi_{[0,a^*_2]}(x_2), \\
\psi_2(x)=-\frac{a_2}{\pi} b_{2}(x_1)\sin(\frac{\pi x_2}{a_2}),
\end{array} 
\end{equation}
where $\chi_{[0,a_1^*]}(\cdot)$ is an indicator function.  
For comparison, we also consider $b_1$, $b_2$ in the form of eigenfunctions as in \cite{meng2022boundary,wang2025sampled} with the corresponding $\psi_1$, $\psi_2$ as follows:
\begin{equation}\label{shapeb1}
 \begin{array}{ll}
 b_{1}(x_1)=\frac{1}{\sqrt{a_1a_2}}\sin(\frac{\pi x_1}{a_1}), b_{2}(x_1)=\frac{1.2}{\sqrt{a_1a_2}}\sin(\frac{2\pi x_1}{a_1}), \\
 \psi_i(x)=-\frac{a_2}{\pi} b_{1}(x_1)\sin(\frac{\pi x_2}{a_2}),~~i=1,2.
\end{array} 
\end{equation}
It is easy to check that for both \eqref{shapeb} and \eqref{shapeb1}, the corresponding ${\bf B}_0 $ satisfies \eqref{eq20dy}.
We fix $\sigma_f=0.1$ ($\sigma_g=0.05$, respectively) and verify LMIs \eqref{eq4343} and \eqref{eq44} for finding the maximum admissible $\sigma_f$ ($\sigma_g$, respectively) that preserves feasibility for the mean-square exponential stability.
Additionally, we verify LMIs \eqref{noisetostate1} for finding the maximum admissible $\sigma_f$ that guarantees the NSS and the NSS in probability. 
The results are summarized in Table \ref{tab1}.
From Table \ref{tab1}, we see that the shape functions \eqref{shapeb} always allow for larger Lipschitz constants than the shape functions \eqref{shapeb1} which are from \cite{meng2022boundary,wang2025sampled}. 
 \begin{table*}[!htb]
    \caption{2D case for $\sigma_f^{\max}$ (maximum admissible $\sigma_f$), $\sigma_g^{\max}$ (maximum admissible $\sigma_g$) from LMIs \eqref{eq4343}, \eqref{eq44} (Theorem \ref{thm1}) and $\sigma_f^{\max}$ from LMIs \eqref{noisetostate1} (Theorem \ref{thm2}): Shape functions \eqref{shapeb} VS \eqref{shapeb1}.}
    \centering
 \begin{tabular}{c|ccc|ccc} \hline

$N$   &  \multicolumn{3}{c|}{\makecell[c]{   Shape functions \eqref{shapeb}  \\ \hline   $\sigma_f=0.1$ ~~ $\sigma_g=0.05$ ~~ Thm \ref{thm2} \\ \hline ~ $\sigma^{\max}_g$ ~~~~~ $\sigma^{\max}_f$ ~~~~~~ $\sigma^{\max}_f$ } } 

    &  \multicolumn{3}{c}{\makecell[c]{  Shape functions \eqref{shapeb1}  \\  \hline  $\sigma_f=0.1$ ~~ $\sigma_g=0.05$ ~~ Thm \ref{thm2} \\ \hline ~ $\sigma^{\max}_g$ ~~~~~ $\sigma^{\max}_f$ ~~~~~~ $\sigma^{\max}_f$ } }  \\ 
    \cline{1-7}

9     & ~~ 0.156   & ~~~ 0.33  & 0.37   &  ~~ 0.035   & ~~~ -- & 0.24\\
11    & ~~ 0.180 & ~~~ 0.40 & 0.43  &  ~~ 0.039   & ~~~ -- & 0.25 \\
13    & ~~ 0.207 & ~~~ 0.48 & 0.52  &  ~~ 0.076   & ~~~ 0.29 & 0.41\\
15    & ~~ 0.211 & ~~~ 0.49 & 0.54  &  ~~ 0.077   & ~~~ 0.30 & 0.42\\
17    & ~~ 0.218 & ~~~ 0.52 & 0.56  &  ~~ 0.080   & ~~~ 0.32  & 0.44    \\
19    & ~~ 0.229 & ~~~ 0.55 & 0.60  &  ~~ 0.086   & ~~~ 0.37 & 0.50\\ \hline
\end{tabular}   \label{tab1}
\end{table*}

For simulations, we consider $f(x,z)=0.33\sin z$ and $g(x,z)=0.05\sin z$ for the multiplicative noise. For the additive noise, we take $g\equiv 1$, $f(x,z)=0.37\sin z$, and $\Sigma(t)\equiv 0.5$, $\Sigma(t)\equiv 1$. 
Take the initial condition $z_0(x)= x_1(\frac{x_1}{a_1}-1)((\frac{x_2}{a_2})^3-(\frac{x_2}{a_2})^2)$ and $N=9$.
The simulation was carried out by using the forward time centered space method and the Euler-Maruyama method with time step $0.00005$ and space step $0.02\times 0.02$.
The simulation results are presented in Fig. \ref{fig1} for the mean-square stability, in Fig. \ref{fig3} for the NSS. The simulation results confirm the theoretical analysis.
Especially, for the exponential stability, the stability in simulations was preserved for larger $\sigma_f= 0.69$ (compared to theoretical value $\sigma_f= 0.33$), which may indicate that our approach is somewhat conservative in this example.
\begin{figure*}[H]
\centerline{\includegraphics[width=10cm]{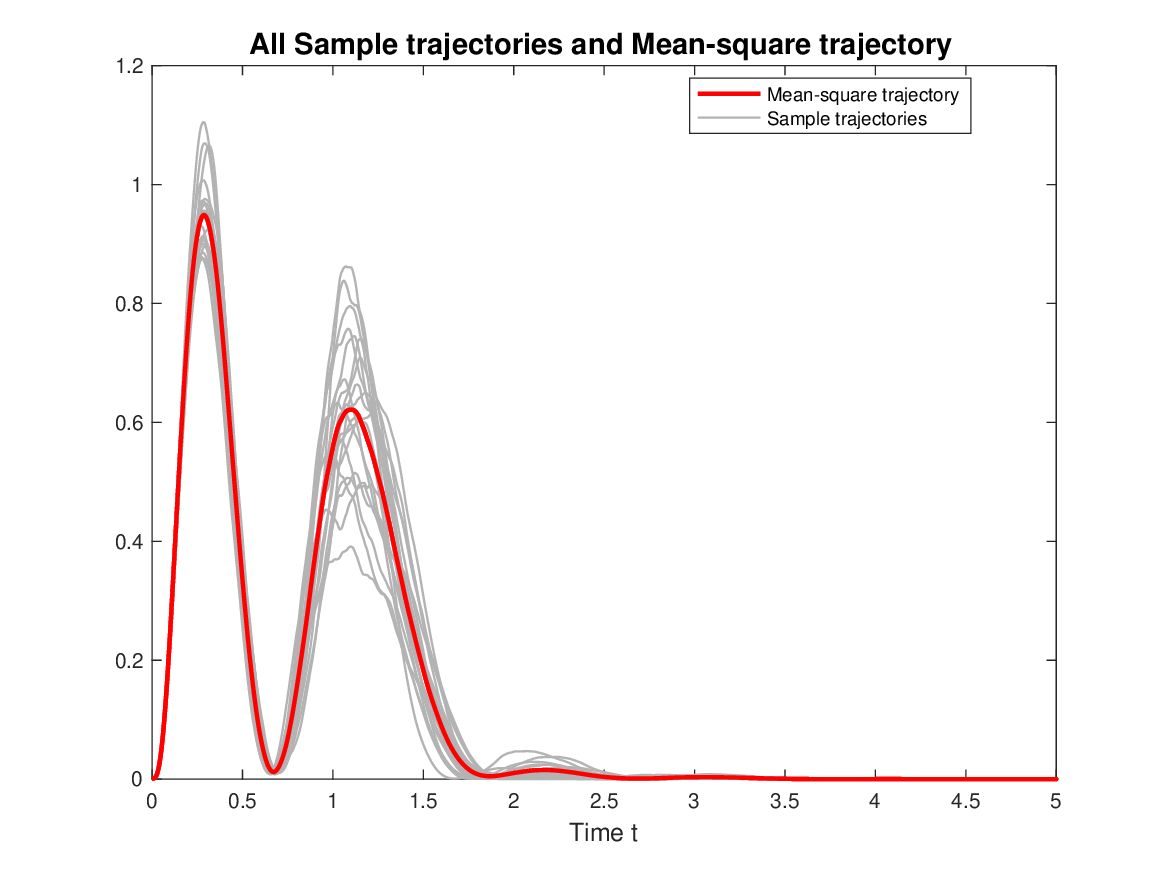}}   
\caption{ The evolution of $\mathbb{E}[\|w(\cdot, t)\|_{L^2}^2+|{\bf u}(t)|^2]$ (The expectation is competed averaging over 20 sample trajectories).}
\label{fig1}
\end{figure*}

\begin{figure*}[H]
\centerline{\includegraphics[width=15cm]{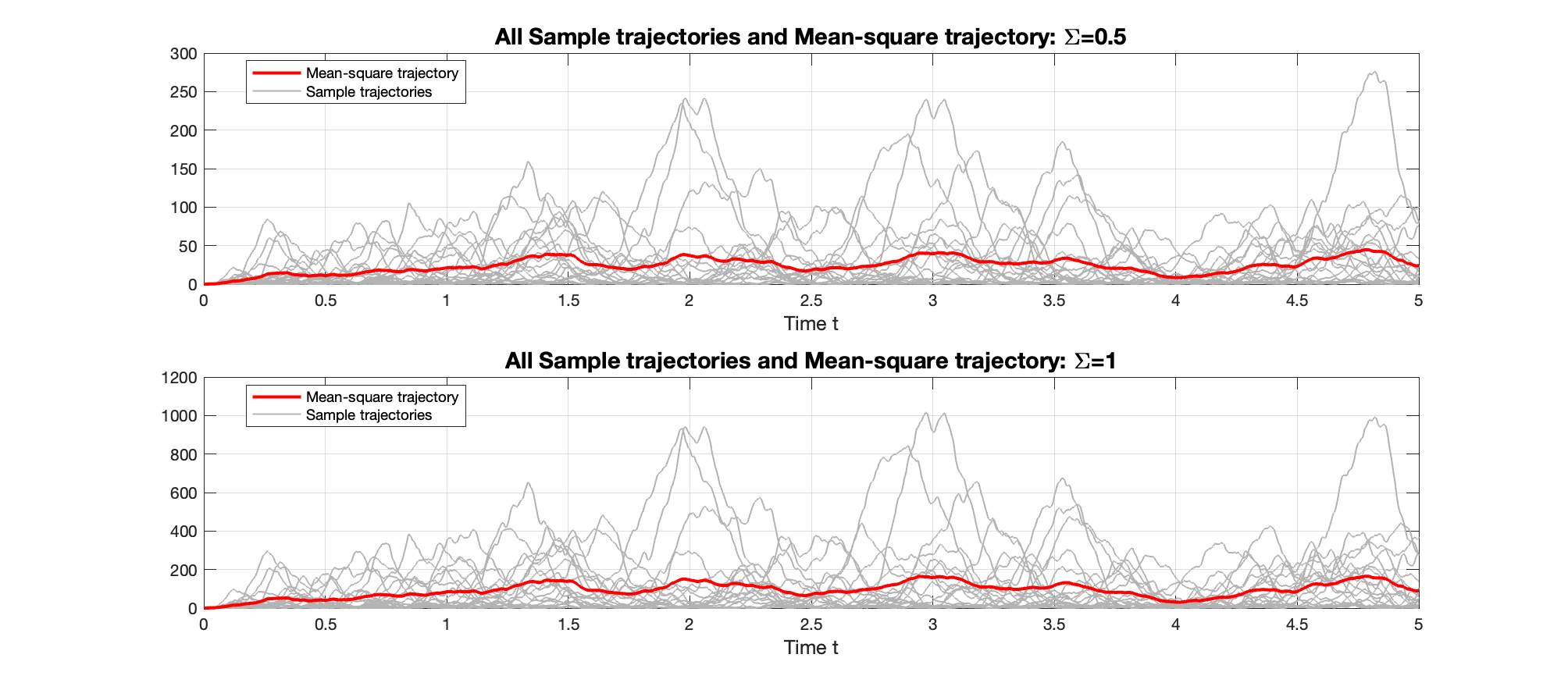}}  
\caption{ The evolution of $\mathbb{E}[\|w(\cdot, t)\|_{L^2}^2+|e^N(t)|^2+|{\bf u}(t)|^2]$ (The expectation is competed averaging over 20 sample trajectories).}
\label{fig3}
\end{figure*}

Next, we consider the 3D case with $\mathcal{O}=[0,a_1]\times [0,a_2]\times [0,a_3]$, where $a_{1}=a_{2}= \sqrt{1.5}$, $a_{3}=1$.
Consider $q=35.37$, which results in an unstable open-loop system with at least $4$ unstable modes ({$\lambda_{1}<\lambda_{2}=\lambda_{3}=\lambda_{4}<q$}).
Choosing $\delta=10^{-3}$, $N_{0}=4$, we have $d=3$. Consider the shape functions in measurement as 
\begin{equation*}
	\begin{array}{ll}
	c_{1}(x_1,x_2)=\frac{1}{\sqrt{a_1a_2a_3}}\sin(\frac{\pi x_1}{a_1})\sin(\frac{\pi x_2}{a_2}), \\
	c_{2}(x_1,x_2)=\frac{0.2}{\sqrt{a_1a_2a_3}}\sin(\frac{2\pi x_1}{a_1})\sin(\frac{\pi x_2}{a_2}), \\
	c_{3}(x_1,x_2)=\frac{0.2}{\sqrt{a_1a_2a_3}}\sin(\frac{\pi x_1}{a_1})\sin(\frac{2\pi x_2}{a_2})	
	\end{array}
\end{equation*}
 and the shape functions $b_i$ in controller and corresponding $\psi_i$, $i=1,2,3$, as follows:
\begin{equation*}
{\small \begin{array}{ll}
 b_{1}(x_1,x_2)=4(2a_1a_2 a_{3})^{-\frac{1}{2}}\sin(\frac{\pi x_1}{a_1})\sin(\frac{\pi x_2}{a_2}), \\
 b_{2}(x_1,x_2)=0.4(2a_1a_2a_{3})^{-\frac{1}{2}}\sin(\frac{2\pi x_1}{a_1})\sin(\frac{\pi x_2}{a_2}),\\
 b_{3}(x_1,x_2)=0.4(2a_1a_2a^*_{3})^{-\frac{1}{2}}\sin(\frac{\pi x_1}{a_1})\sin(\frac{2\pi x_2}{a_2}), \\
 \psi_i(x)=-\frac{a_{3}}{\pi}b_i(x_1,x_2)\sin(\frac{\pi x_3}{a_{3}}), ~i=1,2, \\
 \psi_3(x)=-\frac{a^*_{3}}{\pi}b_3(x_1,x_2)\sin(\frac{\pi x_3}{a^*_{3}})\chi_{[0,a^*_{3}]}(x_3), a^*_{3}=0.8a_3.
\end{array} }
\end{equation*}
Note that if we take $a^*_3=a_3$, the shape functions $b_i$ and corresponding $\psi_i$ are direct extensions of the 1D case in \cite{wang2024SICON,wang2023constructive} and the 2D case in \cite{meng2022boundary,wang2025sampled}. It is easy to check that the corresponding ${\bf C}_0$ and ${\bf B}_0 $ satisfy \eqref{eq20dy}.
The observer gain {$L_0$} is found from \eqref{observergainfulldy} 
and given by 
\begin{equation*}
	L_0={\scriptsize \left[ \begin{array}{cccc}
		62.4643   &      0  &       0 \\
         0 &   6.7736   &      0 \\
          0 &        0  &  6.7736 \\
   -2.5743  &       0   &      0
		\end{array}  \right]}.
\end{equation*}
We set $\sigma_f=0.1$ and verify LMIs \eqref{eq4343} and \eqref{eq44} for finding maximum admissible $\sigma_g$ that guarantees the feasibility. 
The results are presented in Table \ref{table2}. From Table \ref{table2}, we see that choosing $a^*_{3}=0.8a_3$ (where $\psi_3$ is distributed over a subdomain of $\mathcal{O}$) allows for noise intensities more than 3 times larger than the case of $a^*_{3}=a_3$ (where $\psi_3$ is  distributed over the entire $\mathcal{O}$).  
\begin{table*}[!htb]
    \caption{3D case: $\sigma_g^{\max}$ (maximum admissible $\sigma_g$) for $\sigma_f=0.1$: $a^*_3=0.8a_3$ VS  $a^*_3=a_3$.}\label{table2} 
    \centering
 \begin{tabular}{c|cccccccccccccc} \hline
 \centering
 $N$  & 11& 13& 15& 17 & 19 &21 & 23    \\ \cline{1-8}
$\sigma^{\max}_g$ ($a^*_3=0.8a_3$) & 0.142 & 0.142 & 0.149 & 0.162  & 0.165 & 0.165  & 0.171  \\  \hline
$\sigma^{\max}_g$ ($a^*_3=a_3$) & 0.016 & 0.016 & 0.028 & 0.043  & 0.047 & 0.047  & 0.053  \\
\hline
\end{tabular}
\end{table*}

\section{Conclusions}\label{s5}
In this paper, we considered the finite-dimensional boundary observer-based control of high-dimensional stochastic semilinear parabolic PDEs under boundary measurement. 
Improvements and extension of the results to various high-dimensional PDEs may be topics for future research.



\printcredits

\bibliographystyle{cas-model2-names}

\bibliography{cas-refs}

@article{meurer2009trajectory,
  title={Trajectory planning for boundary controlled parabolic {PDE}s with varying parameters on higher-dimensional spatial domains},
  author={Meurer, Thomas and Kugi, Andreas},
  journal={IEEE Transactions on Automatic Control},
  volume={54},
  number={8},
  pages={1854--1868},
  year={2009},
  publisher={IEEE}
}

@article{meurer2011flatness,
  title={Flatness-based trajectory planning for diffusion--reaction systems in a parallelepipedon-{A} spectral approach},
  author={Meurer, Thomas},
  journal={Automatica},
  volume={47},
  number={5},
  pages={935--949},
  year={2011},
  publisher={Elsevier}
}

@article{munteanu2017stabilisation,
  title={Stabilisation of parabolic semilinear equations},
  author={Munteanu, Ionu{\c{t}}},
  journal={International Journal of Control},
  volume={90},
  number={5},
  pages={1063--1076},
  year={2017},
  publisher={Taylor \& Francis}
}

@article{lhachemi2023boundary,
  title={Boundary output feedback stabilisation for 2-{D} and 3-{D} parabolic equations},
  author={Lhachemi, Hugo and Munteanu, Ionut and Prieur, Christophe},
  journal={Automatica},
  volume={176},
  pages={112259},
  year={2025},
  publisher={Elsevier}
}

@article{barbu2011internal,
  title={The internal stabilization by noiseof the linearized Navier-Stokes equation},
  author={Barbu, Viorel},
  journal={ESAIM: Control, Optimisation and Calculus of Variations},
  volume={17},
  number={1},
  pages={117--130},
  year={2011},
  publisher={EDP Sciences}
}

@book{tucsnak2009observation,
  title={Observation and control for operator semigroups},
  author={Tucsnak, Marius and Weiss, George},
  year={2009},
  publisher={Springer Science \& Business Media}
}

@article{meng2022boundary,
  title={Boundary stabilization and observation of a multi-dimensional unstable heat equation},
  author={Meng, Yusen and Feng, Hongyinping},
  journal={arXiv preprint arXiv:2203.12847},
  year={2022}
}

@book{klebaner2012introduction,
  title={Introduction to stochastic calculus with applications},
  author={Klebaner, Fima C},
  year={2005},
  publisher={World Scientific Publishing Company}
}

@book{chow2007stochastic,
  title={Stochastic partial differential equations},
  author={Chow, Pao-Liu},
  year={2007},
  publisher={Chapman and Hall/CRC}
}

@article{katz2020constructive,
  title={Constructive method for finite-dimensional observer-based control of 1-{D} parabolic {PDE}s},
  author={Katz, Rami and Fridman, Emilia},
  journal={Automatica},
  volume={122},
  pages={109285},
  year={2020},
  publisher={Elsevier}
}

@article{wang2023constructive,
  title={Constructive finite-dimensional boundary control of stochastic 1{D} parabolic {PDE}s},
  author={Wang, Pengfei and Katz, Rami and Fridman, Emilia},
  journal={Automatica},
  volume={148},
  pages={110793},
  year={2023},
  publisher={Elsevier}
}

@article{wang2024SICON,
  title={Sampled-data finite-dimensional observer-based control of 1{D} stochastic parabolic {PDE}s},
  author={Wang, Pengfei and Fridman, Emilia},
  journal={SIAM Journal on Control and Optimization},
  volume={62},
  number={1},
  pages={297--325},
  year={2024},
  publisher={SIAM}
}

@article{wang2024auto,
  title={Delayed finite-dimensional observer-based control of 2{D} linear parabolic {PDE}s},
  author={Wang, Pengfei and Fridman, Emilia},
journal={Automatica},
volume={164},
  pages={111607},
  year={2024},
  publisher={Elsevier}
}

@article{barbu2013boundary,
  title={Boundary stabilization of equilibrium solutions to parabolic equations},
  author={Barbu, Viorel},
  journal={IEEE Transactions on Automatic Control},
  volume={58},
  number={9},
  pages={2416--2420},
  year={2013},
  publisher={IEEE}
}

@article{karafyllis2021lyapunov,
  title={Lyapunov-based boundary feedback design for parabolic {PDE}s},
  author={Karafyllis, Iasson},
  journal={International Journal of Control},
  volume={94},
  number={5},
  pages={1247--1260},
  year={2021},
  publisher={Taylor \& Francis}
}

@article{curtain1982finite,
  title={Finite-dimensional compensator design for parabolic distributed systems with point sensors and boundary input},
  author={Curtain, Ruth},
  journal={IEEE Transactions on Automatic Control},
  volume={27},
  number={1},
  pages={98--104},
  year={1982},
  publisher={IEEE}
}

@article{katz2022globalscl,
  title={Global finite-dimensional observer-based stabilization of a semilinear heat equation with large input delay},
  author={Katz, Rami and Fridman, Emilia},
  journal={Systems \& Control Letters},
  volume={165},
  pages={105275},
  year={2022},
  publisher={Elsevier}
}

@article{feng2022boundary,
  title={Boundary stabilization and observation of a weak unstable heat equation in a general multi-dimensional domain},
  author={Feng, Hongyinping and Lang, Pei-Hua and Liu, Jiankang},
  journal={Automatica},
  volume={138},
  pages={110152},
  year={2022},
  publisher={Elsevier}
}

@article{deng2001stabilization,
  title={Stabilization of stochastic nonlinear systems driven by noise of unknown covariance},
  author={Deng, Hua and Krstic, Miroslav and Williams, Ruth J},
  journal={IEEE Transactions on Automatic Control},
  volume={46},
  number={8},
  pages={1237--1253},
  year={2001},
  publisher={IEEE}
}

@article{mateos2014p,
  title={$p$th Moment Noise-to-State Stability of Stochastic Differential Equations with Persistent Noise},
  author={Mateos-Nunez, David and Cortes, Jorge},
  journal={SIAM Journal on Control and Optimization},
  volume={52},
  number={4},
  pages={2399--2421},
  year={2014},
  publisher={SIAM}
}

@article{sakawa1983feedback,
  title={Feedback stabilization of linear diffusion systems},
  author={Sakawa, Yoshiyuki},
  journal={SIAM journal on control and optimization},
  volume={21},
  number={5},
  pages={667--676},
  year={1983},
  publisher={SIAM}
}

@article{balas1988finite,
  title={Finite-dimensional controllers for linear distributed parameter systems: exponential stability using residual mode filters},
  author={Balas, Mark J},
  journal={Journal of Mathematical Analysis and Applications},
  volume={133},
  number={2},
  pages={283--296},
  year={1988},
  publisher={Elsevier}
}

@article{harkort2011finite,
  title={Finite-dimensional observer-based control of linear distributed parameter systems using cascaded output observers},
  author={Harkort, Christian and Deutscher, Joachim},
  journal={International Journal of Control},
  volume={84},
  number={1},
  pages={107--122},
  year={2011},
  publisher={Taylor \& Francis}
}

@article{jadachowski2015backstepping,
  title={Backstepping observers for linear {PDE}s on higher-dimensional spatial domains},
  author={Jadachowski, Lukas and Meurer, Thomas and Kugi, Andreas},
  journal={Automatica},
  volume={51},
  pages={85--97},
  year={2015},
  publisher={Elsevier}
}

@article{vazquez2016explicit,
  title={Explicit output-feedback boundary control of reaction-diffusion {PDE}s on arbitrary-dimensional balls},
  author={Vazquez, Rafael and Krstic, Miroslav},
  journal={ESAIM: Control, Optimisation and Calculus of Variations},
  volume={22},
  number={4},
  pages={1078--1096},
  year={2016}
}

@article{wang2025sampled,
  title={Sampled-data finite-dimensional boundary control of 2D semilinear parabolic stochastic {PDE}s},
  author={Wang, Pengfei and Fridman, Emilia},
  journal={IEEE Transactions on Automatic Control},
  year={2025},
  publisher={IEEE}
}

@inproceedings{pengfei2024CDC,
  title={Finite-dimensional boundary control of 2{D} linear parabolic stochastic {PDE}s  under boundary measurement},
  author={Wang, Pengfei and Fridman, Emilia},
  booktitle={The 63rd IEEE Conference on Decision and Control},
  year={2024},
  organization={IEEE}
}

@article{vazquez2025backstepping,
  title={Backstepping control laws for higher-dimensional PDEs: spatial invariance and domain extension methods},
  author={Vazquez, Rafael},
  journal={IMA Journal of Mathematical Control and Information},
  volume={42},
  number={2},
  pages={dnaf018},
  year={2025},
  publisher={Oxford University Press}
}

@article{djebour2024observer,
  title={Observer-based feedback-control for the stabilization of a class of parabolic systems},
  author={Djebour, Imene Aicha and Ramdani, Karim and Valein, Julie},
  journal={Journal of Optimization Theory and Applications},
  volume={202},
  number={3},
  pages={1217--1241},
  year={2024},
  publisher={Springer}
}

@book{kato2013perturbation,
  title={Perturbation theory for linear operators},
  author={Kato, Tosio},
  volume={132},
  year={2013},
  publisher={Springer Science \& Business Media}
}

@inproceedings{katz2022network,
  title={Network-based deployment of multi-agents without communication of leaders with multiple followers: A PDE approach},
  author={Katz, Rami and Fridman, Emilia and Basre, Idan},
  booktitle={2022 IEEE 61st Conference on Decision and Control (CDC)},
  pages={6089--6096},
  year={2022},
  organization={IEEE}
}

@article{shang2025finite,
  title={Finite-dimensional observer-based boundary control for a one-dimensional stochastic heat equation},
  author={Shang, Yu-Shuo and Wu, Ze-Hao and Zhou, Hua-Cheng},
  journal={Systems \& Control Letters},
  volume={197},
  pages={106046},
  year={2025},
  publisher={Elsevier}
}

@article{Yakubovich1962,
    author = {Yakubovich, V. A.},
    title = {The Solution of Certain Matrix Inequalities in Automatic Control Theory},
    journal = {Soviet Mathematics Doklady},
    volume = {3},
    number = {2},
    pages = {620--623},
    year = {1962},
    note = { (Translation of \emph{Doklady Akademii Nauk SSSR}, Vol. 143, No. 6, pp. 1304--1307, 1962)},
    langid = {english}
}

@article{yakubovich1971s,
    author = {Yakubovich, V. A.},
    title = {S-procedure in Nonlinear Control Theory},
    journal = {Vestnik Leningrad University Mathematics},
    volume = {4},
    pages = {73--93},
    year = {1977},
    note = {(Translation of \emph{Vestnik Leningradskogo Universiteta, ser. 1}, no. 1, pp. 62--77, 1971)},
    langid = {english}
}



\end{document}